\documentclass[lettersize,10pt,journal]{IEEEtran}
\usepackage{amsmath,amsfonts}
\usepackage{algorithmic}
\usepackage{array}
\usepackage[caption=false,font=normal,labelfont=rm,textfont=rm]{subfig}
\usepackage{textcomp}
\usepackage{multirow}
\usepackage{stfloats}
\usepackage{url}
\usepackage{verbatim}
\usepackage{graphicx} 

\usepackage{subfloat}
\usepackage{algorithm, algorithmic}
\usepackage{amssymb}
\usepackage{amsmath,amsfonts,amsthm,bm}
\usepackage{makecell}
\usepackage[justification=centering]{caption}
\usepackage{cite}
\usepackage{soul}
\usepackage{balance}
 \newtheorem{theorem}{Proposition}
\usepackage[colorlinks,linkcolor=magenta,anchorcolor=magenta,citecolor=magenta,bookmarks=true]{hyperref}  
 
\usepackage{epstopdf,amsthm,stfloats,siunitx,amssymb,wasysym,algorithm,algorithmic,array,url, color} 
  
\begin{document}
\title{Range-Angle Estimation for FDA-MIMO System With Frequency Offset}
\author{Mengjiang~Sun,~\IEEEmembership{Student~Member,~IEEE,} Peng~Chen,~\IEEEmembership{Senior~Member,~IEEE,} Zhenxin~Cao,~\IEEEmembership{Member,~IEEE} 
\thanks{This work was supported in part by the Natural Science Foundation for Excellent Young Scholars of Jiangsu Province under Grant BK20220128, the Open Fund of State Key Laboratory of Integrated Chips and Systems under Grant SKLICS-K202305, the Open Fund of National Key Laboratory of Wireless Communications Foundation under Grant IFN20230105, the Open Fund of National Key Laboratory on Electromagnetic Environmental Effects and Electro-optical Engineering under Grant JCKYS2023LD6, the Open Fund of ISN State Key Lab under Grant ISN24-04, and the National Natural Science Foundation of China under Grant 61801112.} 
\thanks{Mengjiang~Sun, Peng~Chen and Zhenxin~Cao are with the State Key Laboratory of Millimeter Waves, Southeast University, Nanjing 210096, China (e-mail: \{mengjiangsun, chenpengseu, caozx\}@seu.edu.cn). Peng~Chen is also with State Key Laboratory of Integrated Chips and Systems, Fudan University, Shanghai 201203, China.}   \thanks{\textit{(Corresponding author: Peng~Chen)}}}

\markboth{IEEE TRANSACTIONS ON AEROSPACE AND ELECTRONIC SYSTEMS}%
{Shell \MakeLowercase{\textit{et al.}}: Bare Demo of IEEEtran.cls for IEEE Journals}
 
\maketitle

\begin{abstract}
Frequency diverse array multiple-input multiple-output (FDA-MIMO) radar differs from the traditional phased array (PA) radar, and can form range-angle-dependent beampattern and differentiate between closely spaced targets sharing the same angle but occupying distinct range cells. In the FDA-MIMO radar, target range estimation is achieved by employing a subtle frequency variation between adjacent array antennas, so the estimation performance is degraded severely in a practical scenario with frequency offset. In this paper, the range-angle estimation problem for FDA-MIMO radar is considered with frequency offsets in both transmitting and receiving arrays. First, we build a system model for the FDA-MIMO radar with transmitting and receiving frequency offsets. Then, the frequency offset is transferred into an equalized additional noise. The noise characteristics are analyzed in detail theoretically, together with the influence on the range-angle estimation. Moreover, since the effect of the transmitting frequency offset is similar to additional colored noise, denoising algorithms are introduced to mitigate the performance deterioration caused by the frequency offset. Finally, Cram\'{e}r-Rao lower bounds (CRLB) for the range-angle estimation are derived in the scenario with the frequency offsets. Simulation results show the analysis of frequency offset and the corresponding estimation performance using different algorithms. 
\end{abstract}


\section{Introduction}
\IEEEPARstart{i}{n} recent years, frequency diverse array (FDA) radar has been studied in different fields, including target sensing~\cite{sensing}, mainlobe interference suppression~\cite{mainlobe}, integrated sensing and communication (ISAC)~\cite{CAESAR} and etc. FDA radar was introduced in~\cite{FDA} as a technique to achieve beamforming which is dependent on both range and angle, thus effectively mitigating range-ambiguous clutter.~\cite{FDA_clutter}. Unlike traditional phased array (PA) radars, where all array antennas operate at the same frequency, the FDA radar introduces a frequency difference between array antennas to obtain range-angle dependent beampattern. Due to the inherent coupling of range and angle in the FDA beampattern, it is not feasible to directly estimate the targets' range and angle from the peaks in the beamforming output. Therefore, there are extensive literature dealing with the decoupling methods~\cite{log_increment,random_increment,double_pulse,FDA_sbarray,FDA_MIMO,FDA_MIMO2}. For example, random and logarithmic frequency increments were proposed to obtain decoupled beampattern in~\cite{log_increment,random_increment}, but the beampattern has a high sidelobe and poor beamforming performance. In a double-pulse FDA radar system, two distinct pulse types are transmitted to obtain the targets' angles and ranges separately~\cite{double_pulse}, one characterized by a zero frequency increment and the other by a non-zero frequency increment. In ref.~\cite{FDA_MIMO,LAN1,LAN2}, multiple-input multiple-output (MIMO) and FDA are combined as a frequency diverse array multiple-input multiple-output (FDA-MIMO) radar. Ambiguity can be mitigated with increased degrees of freedom from the MIMO technique. Additionally, ref.~\cite{FDA_sbarray} introduced a subarray-based FDA radar system that employs two distinct frequency increments, allowing for the direct estimation of both range and angle from the peaks in the beamforming output.

In the FDA-MIMO radar, the frequency difference between adjacent array antennas ensures signal orthogonality, facilitating their separation through matched filters in the receiver. Hence, the range and angle estimations are decoupled, and a joint range-angle estimation becomes possible. Ref.~\cite{FDA_CRLB} studies the Cram\'{e}r-Rao lower bound (CRLB) of range-angle estimation in the FDA-MIMO radar, together with the estimation resolution. Subspace-based methods such as two-dimensional multiple signal classification (2D-MUSIC) address the range-angle estimation problem~\cite{FDA_CRLB}. In ref.~\cite{FDA_MIMO_MUSIC}, a reduced-dimension MUSIC algorithm is further proposed, using multiple one-dimensional MUSIC (1D-MUSIC) steps to avoid the two-dimensional search, so the computational complexity is greatly reduced. Moreover, compressed sensing (CS)-based methods are also proposed for the angle-range estimation in the FDA-MIMO radar. For example, grid-based CS algorithms like orthogonal matching pursuit (OMP) algorithm~\cite{OMP} can be used in the range-angle estimation in FDA-MIMO radar system through 2D searching, but the estimation performance degradation caused by grid-mismatch and high computational complexity are hard to balance~\cite{FDA_MIMO_CS,CS_mismatch}. Hence, an atomic norm minimization (ANM)-based method that utilizes the signal sparsity in a continuous parameter domain is proposed and can overcome the grid-mismatch problem~\cite{ANM}. 2D-ANM estimation method is also proposed by extending to a two-fold ANM~\cite{ANM_2d}, but a high dimension brings a huge computation burden. To overcome this problem, a decoupled ANM is also proposed for the joint estimation~\cite{DANM}.

In the traditional PA radars, the direction of arrival (DOA) is estimated through a steering vector concerning carrier frequency and DOA, where transmitting frequency offset can be negligible since it is significantly smaller than the carrier frequency. However, in the FDA-MIMO radar, the frequency difference is also substantially smaller than the carrier frequency, so the steering vector deviation induced by the frequency offset is serious, degrading the performance of the DOA estimation. In ref.~\cite{ANM_gp_error}, the DOA estimation problem takes into account gain and phase errors within the array antennas. Similarly, ref.~\cite{PN} shows that the phase noise in the frequency-modulated continuous wave (FMCW) radar also leads to the steering vector deviation. Therefore, with the frequency offset, the range estimation performance in the FDA-MIMO system will be degraded severely. Ref.~\cite{fda_error} investigated the mainlobe offset and signal to interference and noise ratio (SINR) resulting from frequency increment offset in FDA-MIMO radar. However, only the deterioration of performance of MUSIC algorithm is included and the frequency increment offset is assumed to be unchanged in different pulses. A comprehensive theoretical performance analysis of the influence of frequency offsets including transmitting and receiving frequency offsets in the FDA-MIMO radar is necessary. However, relevant studies have not been made before.

Moreover, the existing algorithms for the joint range and angle estimation in the FDA-MIMO radar are based on white noise assumption~\cite{FDA_MIMO_CS,ANM_2d,DANM,DANM2,FDA_MIMO_MUSIC}. However, according to our analysis, the influence of transmitting and receiving frequency offsets cannot be equalized as white noise. Instead, the additional noise caused by frequency offsets can be colored. Large numbers of literature show that colored noise causes severe degradation in the estimation performance, especially for the subspace-based methods~\cite{cn3,cn5,cn4}. Therefore, various algorithms for scenes with colored noise are proposed~\cite{transformation,ML,rotate,four_level}. When perfectly known, the colored noise can be directly whitened through a linear transformation~\cite{whitening_transformation}. Estimation algorithms can be applied without complication when factoring in the transformation step~\cite{transformation,ML}. Additionally, the covariance differencing method can mitigate the influence of colored noise through rotating arrays~\cite{rotate}. High-order statistics can also suppress the Gaussian process component~\cite{c4_1,c4_2,four_level}. Ref.~\cite{four_level} uses the fourth-order cumulant to avoid the influence of colored noise. However, the size of the fourth-order cumulant matrix is large, and the high computational burden prevents it from practical application.

In this paper, the range-angle estimation problem for the FDA-MIMO radar is considered with frequency offsets in both transmitting and receiving arrays. The main contributions of this paper are summarized below:
\begin{itemize}
	\item \textbf{The system model of FDA-MIMO radar with transmitting and receiving frequency offsets}: Considering transmitting and receiving frequency offsets, a novel model of FDA-MIMO radar is formulated. The frequency offsets are transferred to an additional noise so that the influence of the frequency offsets can be shown theoretically.
	\item \textbf{The influence of the frequency offsets on range-angle estimation in the FDA-MIMO radar}: We show the deterioration of the range-angle estimation caused by the frequency offsets by analyzing the characteristics of the equalized additional noise theoretically.
    \item \textbf{The influence of the frequency offsets on the performance of the joint range-angle estimation and the corresponding CRLB}: The performance of joint range-angle estimation algorithms in the FDA-MIMO radar with the frequency offsets are shown. Specifically, the analysis shows that the equivalent noises are colored, and thus algorithms for colored noise mitigation are given. CRLB for the joint estimation is derived with the frequency offsets.
\end{itemize}

The rest of this paper is organized as follows. In Section~\ref{sec2}, we build the system model for the FDA-MIMO radar with transmitting and receiving frequency offsets, and the characteristics of the offsets are also analyzed. In Section~\ref{sec4}, denoising algorithms for the frequency offset are presented. The CRLB for the range-angle estimation is derived in Section~\ref{sec3}. Simulation results and analysis are shown in Section~\ref{sec5}. Finally, Section~\ref{sec6} concludes this paper.

\emph{Notations:} $\boldsymbol{x}^\text{T}$, $\boldsymbol{x}^\text{H}$ and $\boldsymbol{x}^\text{*}$ denote the transpose of $\boldsymbol{x}$, the Hermitian transpose of $\boldsymbol{x}$ and the conjugation of $\boldsymbol{x}$, respectively. $\| \boldsymbol{x} \|_1$ and $\| \boldsymbol{x} \|_2$ denote the $\ell_1$ norm of $\boldsymbol{x}$ and the $\ell_2$ norm of $\boldsymbol{x}$, respectively. $\mathrm{Tr} \{ \boldsymbol{X} \}$ denotes the trace of $\boldsymbol{x}$. $\mathbb{E} \{ \boldsymbol{x} \}$ denotes the expectation of $\boldsymbol{x}$. $\mathcal{R} \{a\}$ denotes the real part of complex value $a$. $\otimes$ denotes the Kronecker product. $\odot$ denotes the Khatri-Rao product. The boldface capital letters represent matrices, as indicated by $\boldsymbol{X}$, while lowercase letters represent vectors, denoted by $\boldsymbol{x}$.

\section{The System Model of FDA-MIMO Radar}\label{sec2}
We consider a co-located FDA-MIMO radar comprising a uniform linear array (ULA) equipped with $N$ transmitting antennas and a ULA with $M$ receiving antennas as shown in Fig.~\ref{fig1}.

\subsection{Transmitting Signal Model}
{Use the first array antenna as a reference, and the signal emitted from the $n$-th transmitting antenna is} \cite{sensing,random_increment,E_N}
\begin{equation}
s_n(t)=\sqrt{\frac{E}{N}}  \mathrm{\Pi}  \left( \frac{t}{T_p} \right) e^{-j2 \pi f_{\mathrm{t},n} t}, n=1,2, \cdots, N,
\end{equation}
where $E$ is the total transmitted energy. $\mathrm{\Pi}(t)$ is the rectangular pulse with value $1$ for $t \in [0,1) $ and zero otherwise. $T_p$ is the pulse duration.  $f_{\mathrm{t},n}$ is the transmitting carrier frequency at the $n$-th antenna, and can be expressed as
\begin{equation}
f_{\mathrm{t},n} =  f_0 + \left( n-1 \right) \Delta f +f_{e,\mathrm{t},n},
\end{equation}
 where $f_0$ represents the carrier frequency, and $ \Delta f$ is the frequency difference between adjacent antennas. Note that $ \Delta f$ is significantly smaller than the carrier frequency $f_0$, i.e., $ f_0 \gg \Delta f$. $f_{e,\mathrm{t},n}$ denotes the transmitting frequency offset at the $n$-th antenna. According to~\cite{CFO_tr,CFO_0,fet10}, the carrier frequency offset is assumed to obey Gaussian distribution. Since the pulse duration is much less than the pulse repetition time, we assume that the frequency offsets remain constant for the duration of one pulse but vary and are independent and identically distributed (i.i.d) among different pulses.
 \begin{figure}[H]
\centering
\fontsize{10pt}{12pt}\selectfont
\subfloat[Transmitter]{
\label{$ f_s = 1e^7$}
\includegraphics[width=3.1in]{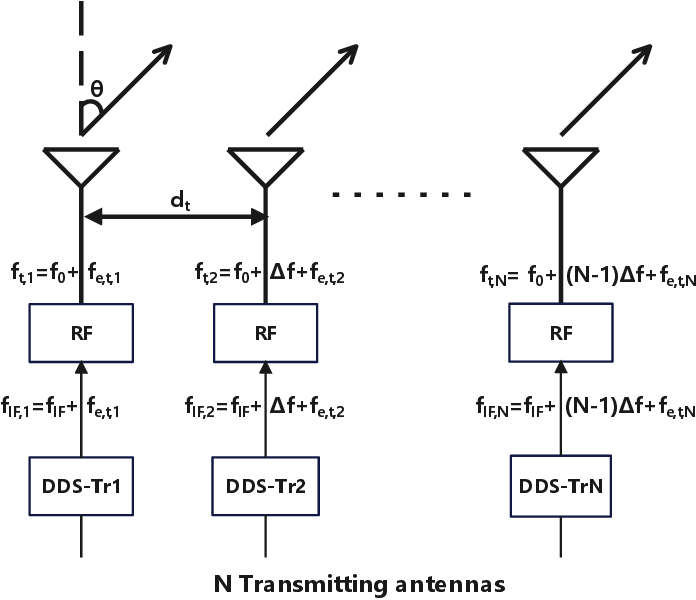}}
\quad
\subfloat[Receiver]{
\label{fig:subfig:c}  
\includegraphics[width=3in]{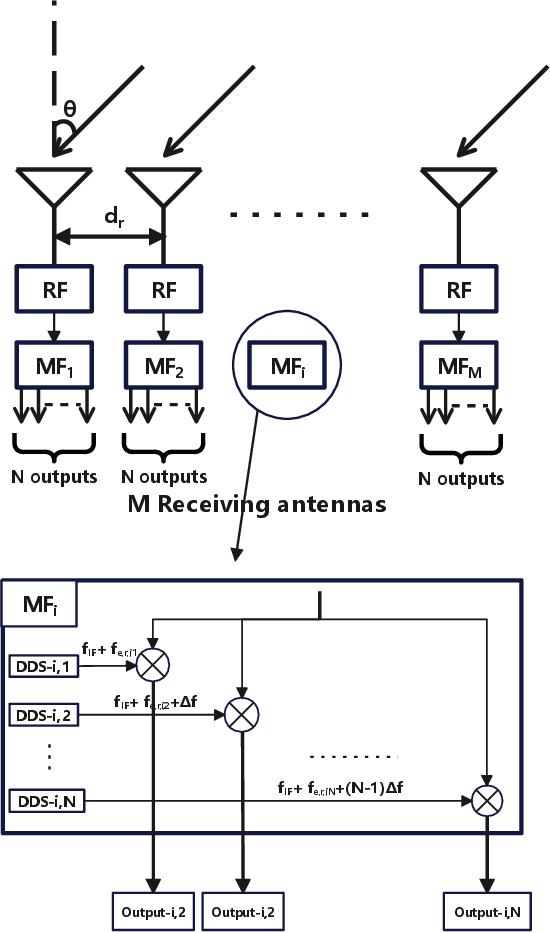}}
\caption{Diagram of an FDA-MIMO radar.}
\label{fig1}
\end{figure}

{The $N$ transmitted waveforms $\{s_n(t)\}_{n=1,2, \cdots, N}$ are supposed to be orthogonal to each other} \cite{FDA_CRLB,sensing,FDA_orth}{, namely,}
\begin{equation} 
\int_0^{T_P} s_{n_1}(t) s_{n_2}^\text{*}(t) dt = 0, n_1 \neq n_2, 
\end{equation}
where $n_1,n_2 \in [ 1,2, \cdots, N ].$ After simplification, the orthogonality equals
\begin{equation} \label{orth0}
\int_0^{T_P} e^{j 2 \pi \Delta f (n_1-n_2) t} dt= 0, n_1 \neq n_2,
\end{equation}
where $T_p = 1/ \Delta f$ satisfies the orthogonality.

\subsection{Receiving Signal Model}
Suppose a  far-field point target reflects the transmitting signals with direction $\theta$ and range $r$, and the round-trip propagation delay from the $n$-th transmitting antenna to the $m$-th receiving antenna is
\begin{equation}
\tau_{m,n} = \frac{1}{c}\left[2 r-(n-1)d_t \sin(\theta) - (m-1)d_r \sin(\theta)\right],
\end{equation}
where $d_t$  and $d_r$ are the element spacing of the transmitting array and receiving array, respectively. $d_t$ and $d_r$ are assumed to equal half of the wavelength, i.e., $d_t = d_r = d =c/(2f_0)$. The received signal at the $m$-th antenna can be represented as
\begin{equation}
\begin{split}
{Y}_{m}(t) &= \sum_{i=1}^N \alpha s_i(t-\tau_{m,i}) \\ &=\sum_{i=1}^{N} \alpha \sqrt{\frac{E}{N}} \mathrm{\Pi} (\frac{t-\tau_{m,i}}{T_p}) e^{-j2 \pi f_{t,i} (t-\tau_{m,i})},
\end{split}
\end{equation}
where $\alpha$ is the coefficient of the received signal  in the $l$-th pulse, accounting for the transmit amplitude, target reflecting coefficient, propogation decay, etc. The received signal is then matched filtered by $e^{j 2 \pi f_{r,m,n} t}$ to separate signals from different transmitting antennas, where $f_{r,m,n}$ denotes the $n$-th down-conversion frequency for the $m$-th receiving antenna. $f_{r,m,n}$ satisfies
\begin{equation}
f_{r,m,n} = f_0 + (n-1) \Delta f +f_{e,r,m,n}, m=1,2, \cdots, M,
\end{equation}
where $f_{e,r,m,n}$ denotes the receiving frequency offset of the $n$-th down-conversion frequency for the $m$-th receiving antenna, which is assumed to be mainly induced by oscillator frequency errors. Ref.~\cite{CFO_r} builds the model for oscillator frequency errors, which includes white frequency noise, Flicker frequency noise and random walking frequency noise. According to the frequency noise spectrum in ref.~\cite{CFO_r}, white frequency noise dominates the frequency noise at high working frequency. Therefore, we build the receiving frequency offset as white noise, which obeys i.i.d Gaussian distribution for each receiving antenna in each pulse. The receiving frequency offsets at each antenna are assumed to remain constant for the duration of one pulse but vary and are i.i.d among different pulses. {Moreover, the transmitting frequency offsets are assumed to be independent of the receiving frequency offsets in the same pulse since distinct antenna elements and different Direct Digital Synthesis (DDS) are used for the transmitting and receiving arrays, respectively.} The received signal at the $m$-th antenna is then fed into $N$ matched filters, and the $n$-th output of the $m$-th received antenna is
\begin{equation}
\begin{split}
{Y}_{m,n} &= \int_{\tau_{m,i}}^{\tau_{m,i}+T_p} Y_{m}(t) e^{j 2 \pi f_{r,m,n} t}dt \\ &= \sum_{i=1}^{N} \int_{\tau_{m,i}}^{\tau_{m,i}+T_p} \alpha \sqrt{\frac{E}{N}}  e^{-j 2 \pi f_{t,i} (t-\tau_{m,i})} e^{j 2 \pi f_{r,m,n} t} dt \\ &=  \sum_{i=1}^{N}  \alpha \sqrt{\frac{E}{N}} e^{j 2 \pi f_{r,m,n} \tau_{m,i}} \int_{0}^{T_p} e^{-j 2 \pi (f_{t,i} -f_{r,m,n})t} dt.
\end{split}\label{eq0}
\end{equation}
{Since the frequency difference $\Delta f$ and the receiving frequency offset $f_{e,r,m,n}$ are both far less than the carrier frequency $f_c$, the differences caused by $\Delta f$ and $f_{e,r,m,n}$ are ignored within the array aperture} \cite{sensing,FDA_orth,FDA_MIMO}{, i.e.}
\begin{equation}\label{asu1}
			\begin{aligned}
				\frac{\Delta f d \sin \theta}{c} & = \frac{\Delta f\sin \theta}{2 f_c} \ll 1
				\\ \frac{f_{e,r,m,n} d \sin \theta}{c} & = \frac{f_{e,r,m,n} \sin \theta}{2 f_c} \ll 1
			\end{aligned}
		\end{equation}
$e^{j 2 \pi f_{r,m,n} \tau_{m,i}}$ is then simplified as 
\begin{equation}
\begin{split}
& \qquad e^{j 2 \pi f_{r,m,n} \tau_{m,i} }
\\& = e^{j 2 \pi ( f_0 + (n-1) \Delta f +f_{e,r,m,n})(\frac{2 r-(m+i-2)d \sin(\theta)}{c})}
\\& = e^{j 2 \pi [ (f_0 + (n-1) \Delta f +f_{e,r,m,n})\frac{2r}{c} - f_0  \frac{(m+i-2)d \sin(\theta)}{c}]}.
\end{split}
\end{equation}
Both the transmitting and receiving frequency offsets are assumed to be far less than the frequency resolution, i.e., $\{\left |f_{e,t,i} \right | ,\left |f_{e,r,m,n}\right | \} \ll 1/T_p  = \Delta f$. {The target is assumed to be located in the unambiguous range, i.e., $r < \frac{c}{2\Delta f}$.} We can then simplify $e^{j 2 \pi f_{e,r,m,n}\frac{2r}{c}}$ and $e^{-j 2 \pi (f_{e,t,i}-f_{e,r,m,n})t}$ through the first-order Taylor expansion. {The filter output $Y_{m,n}$ is firstly shown as}
			\begin{equation} \label{Ymn0}
				\begin{aligned}
					& {Y}_{m,n}  
     \\ = &  \sum_{i=1}^{N}  \alpha \sqrt{\frac{E}{N}} e^{j 2 \pi f_{r,m,n} \tau_{m,i}} \int_{0}^{T_p} e^{-j 2 \pi (f_{tr,i} -f_{r,m,n})t} dt
     \\ &+ N_0
					\\ = & \sum_{i=1}^{N}  \alpha \sqrt{\frac{E}{N}}  e^{j 2 \pi [ (f_0 + (n-1) \Delta f +f_{e,r,m,n})\frac{2r}{c} - f_0  \frac{(m+i-2)d \sin(\theta)}{c}]}  
     \\& \int_{0}^{T_p} e^{-j 2 \pi [(i-n) \Delta f + f_{e,t,i}-f_{e,r,m,n} ] t} dt  + N_0
				\end{aligned}
			\end{equation}
			{$N_0$ denotes the additive Gaussian white noise. The carrier frequency offset $f_{e,t,i}$ and the receiving frequency offset $f_{e,r,m,n}$ are both assumed to be smaller than the frequency increment $\Delta f$, and the pulse duration $T_p$ is set as $\frac{1}{\Delta f}$, we have following approximations using the first-order Taylor expansion}
			
			\begin{equation}\label{taylor}
				\begin{aligned}
					e^{j 2 \pi f_{e,r,m,n}\frac{2r}{c}} & \approx 1 + j 2 \pi f_{e,r,m,n}\frac{2r}{c},
					\\ e^{-j 2 \pi (f_{e,t,i}-f_{e,r,m,n} ) t} & \approx  1- j 2 \pi (f_{e,t,i}-f_{e,r,m,n} ) t.
				\end{aligned}
			\end{equation}
			{The filter output $Y_{m,n}$ is then simplified as}
			\begin{equation} 
			\begin{split}
				{Y}_{m,n}  = & \; \sum_{i=1}^{N}  \alpha \sqrt{\frac{E}{N}} e^{j 2 \pi [ (f_0 + (n-1) \Delta f )\frac{2r}{c}- f_0  \frac{(m+i-2)d \sin(\theta)}{c}]}
				\\ & ( 1+j 2 \pi f_{e,r,m,n} \frac{2r}{c}) 
				\int_{0}^{T_p} \bigg [ e^{-j 2 \pi (i-n) \Delta f t} 
                \\ & (1- j 2 \pi 
				(f_{e,t,i}-f_{e,r,m,n})t) \bigg ]dt + N_0.
				\\ = & \; \alpha \sqrt{\frac{E}{N}} e^{j 4 \pi f_0 \frac{r}{c}} e^{j 2 \pi (n-1) \Delta f \frac{2r}{c}} ( 1+j 2 \pi f_{e,r,m,n} \frac{2r}{c}) 
                \\ & \; \sum_{i=1}^{N} e^{-j 2 \pi f_0  \frac{(m+i-2)d \sin(\theta)}{c}}
				\quad \bigg[ \int_{0}^{T_p} e^{-j 2 \pi (i-n) \Delta f t} dt - 
                \\ & j 2 \pi (f_{e,t,i}-f_{e,r,m,n}) \int_{0}^{T_p} e^{-j 2 \pi (i-n) \Delta f t} t dt \bigg] + N_0. 
			\end{split}
		      \end{equation}
        {According to the orthogonality as shown in Eq.(}\ref{orth0}{)  in the manuscript, we have}
        \begin{equation}
			\int_{0}^{T_p} e^{-j 2 \pi (i-n) \Delta f t} dt = \left\{
			\begin{aligned}
				T_p, & \; i = n \\
				0, & \; i \neq n 
			\end{aligned}
			\right.
		\end{equation}
        {Therefore, We have}
        \begin{equation}
        \begin{aligned}
			& \; \sum_{i=1}^{N} e^{-j 2 \pi f_0  \frac{(m+i-2)d \sin(\theta)}{c}} \int_{0}^{T_p} e^{-j 2 \pi (i-n) \Delta f t} dt = 
   \\ & \; T_p e^{-j 2 \pi f_0  \frac{(m+n-2)d \sin(\theta)}{c}}
        \end{aligned}
		\end{equation}
        {Denote $\beta$ and $I_{1,i,n}$ as}
        \begin{equation}
			\begin{aligned}
				\beta & = \alpha T_p \sqrt{\frac{E}{N}} e^{j 4 \pi f_0 \frac{r}{c}},
				\\ I_{1,i,n} & = \int_{0}^{T_p} e^{-j 2 \pi (i-n) \Delta f t} t dt.
			\end{aligned}
		\end{equation}
        {The filter output $Y_{m,n}$ is then simplified as}
        \begin{equation}
			\begin{aligned}
				Y_{m,n} = & \beta e^{j 2 \pi (n-1) \Delta f \frac{2r}{c}} ( 1+j 2 \pi f_{e,r,m,n} \frac{2r}{c})
                \\ & \bigg[ e^{-j 2 \pi f_0  \frac{(m+n-2)d \sin(\theta)}{c}} -  \frac{j 2 \pi}{T_p} \sum_{i=1}^{N}(f_{e,t,i}-f_{e,r,m,n})
				\\ &  e^{-j 2 \pi f_0  \frac{(m+i-2)d \sin(\theta)}{c}} I_{1,i,n} \bigg] + N_0.
			\end{aligned}
		\end{equation}
        {Since the maximum unambiguous range of FDA-MIMO radar is $\frac{c}{2\Delta f}$, we have}
        \begin{equation}
			f_{e,r,m,n}\frac{2r}{c} \leq f_{e,r,m,n}\frac{2c}{2c\Delta f} = \frac{f_{e,r,m,n}}{\Delta f}.
		\end{equation}
        {According to the assumption, the frequency offset $f_{e,r,m,n}$ is assumed to be far less than the frequency difference $\Delta f$, and thus $2 \pi f_{e,r,m,n} \frac{2r}{c}$ is a small quantity compared with $1$. Moreover, $I_{1,i,n}$ satisfies }
        \begin{equation}
			\left| I_{1,i,n} \right| \leq \int_{0}^{T_p} \left|e^{-j 2 \pi (i-n) \Delta f t} t\right| dt = \frac{T_p^2}{2}.
		\end{equation}
        {With the orthogonality requirement $T_p = \frac{1}{\Delta f}$, we then have}
        \begin{equation}
        \begin{aligned}
		  & \frac{\left|(f_{e,t,i}-f_{e,r,m,n}) e^{-j 2 \pi f_0  \frac{(m+i-2)d \sin(\theta)}{c}} I_{1,i,n} \right|}{\left|T_p e^{-j 2 \pi f_0  \frac{(m+n-2)d \sin(\theta)}{c}}\right|} \leq 
            \\ & \frac{\left|(f_{e,t,i}-f_{e,r,m,n})\right| T_p^2}{2 T_p} = \frac{\left|f_{e,t,i}-f_{e,r,m,n}\right|}{2\Delta f}
        \end{aligned}
		\end{equation}
        {Since the frequency offset $f_{e,r,m,n}$ and $f_{e,t,i}$ are assumed to be far less than the frequency increment $\Delta f$, $\left|(f_{e,t,i}-f_{e,r,m,n}) e^{-j 2 \pi f_0  \frac{(m+i-2)d \sin(\theta)}{c}} I_{1,i,n} \right|$ is a small quantity compared with $\left|T_p e^{-j 2 \pi f_0  \frac{(m+n-2)d \sin(\theta)}{c}}\right|$. Therefore, the product of the two first-order quantities, i.e., $2 \pi f_{e,r,m,n} \frac{2r}{c} \frac{2 \pi}{T_p} \sum_{i=1}^{N}(f_{e,t,i}-f_{e,r,m,n}) $ $ e^{-j 2 \pi f_0  \frac{(m+i-2)d \sin(\theta)}{c}} I_{1,i,n}$, is a second-order small quantity compared with $e^{-j 2 \pi f_0  \frac{(m+n-2)d \sin(\theta)}{c}}$. We retain the first-order small quantities and omit the second-order small quantity, and $Y_{m,n}$ is further written as }
        \begin{equation}
			\begin{aligned}
				Y_{m,n} = & \beta e^{j 2 \pi (n-1) \Delta f \frac{2r}{c}} \bigg[ e^{-j 2 \pi f_0  \frac{(m+n-2)d \sin(\theta)}{c}} + 
                \\ & j 2 \pi f_{e,r,m,n} \frac{2r}{c} e^{-j 2 \pi f_0  \frac{(m+n-2)d \sin(\theta)}{c}} - \frac{j 2 \pi}{T_p}
				\\ &  \sum_{i=1}^{N}(f_{e,t,i}-f_{e,r,m,n})e^{-j 2 \pi f_0  \frac{(m+i-2)d \sin(\theta)}{c}} I_{1,i,n}\bigg] + N_0
				.
			\end{aligned}
		\end{equation}
        {We denote $N_{t,m,n}$ and $N_{r,m,n}$ as} 
			\begin{equation}
				\begin{split}
					& N_{t,m,n} = 
					j 2 \pi \frac{\beta}{T_p} \sum_{i=1}^{N} f_{e,t,i}  e^{j 2 \pi [ (n-1) \Delta f \frac{2r}{c}- (m+i-2)f_\theta]} I_{1,i,n},
				\end{split}\label{eq4}
			\end{equation}
			\begin{equation}
				\begin{split}
					 N_{r,m,n} =& 
					j 2 \pi f_{e,r,m,n}\frac{2 r}{c} \beta  e^{j 2 \pi [(n-1) \Delta f \frac{2r}{c}- (m+n-2) f_\theta}] +
					\\ &  j 2 \pi f_{e,r,m,n} \frac{\beta}{T_p} \sum_{i=1}^{N}  e^{j 2 \pi [ (n-1) \Delta f \frac{2r}{c}- (m+i-2)f_\theta]} I_{1,i,n},
				\end{split}\label{eq5}
			\end{equation}
	{where $f_\theta \triangleq f_0 d \sin (\theta)/c$ is the normalized spatial frequency. $Y_{m,n}$ is finally written as} 
		\begin{equation} \label{Ymn}
        \begin{aligned}
			Y_{m,n} = & \beta e^{j 2 \pi (n-1) \Delta f \frac{2r}{c}} e^{-j 2 \pi f_0  \frac{(m+n-2)d \sin(\theta)}{c}}
            \\ & + N_{t,m,n} + N_{r,m,n} + N_0.
        \end{aligned}
		\end{equation}
$\boldsymbol{Y}$ can be then represented as 
\begin{equation} \label{eq16}
\boldsymbol{Y} = \beta \boldsymbol{a}_r (\theta) \boldsymbol{a}_t (\theta,r)^\text{T} +\boldsymbol{N}_t + \boldsymbol{N}_r + \boldsymbol{N}_0,
\end{equation}
The receiving steering vector is
\begin{equation}
\boldsymbol{a}_r (\theta) = [1,e^{-j 2 \pi f_\theta},\cdots,e^{-j 2 \pi(M-1)f_\theta}]^\text{T}.
\end{equation}
The transmitting steering vector is
\begin{equation}
\boldsymbol{a}_t (\theta,r) = [1,e^{j 2 \pi \phi},\cdots,e^{j 2 \pi ((N-1)\phi)}]^\text{T},
\end{equation}
with $ \phi = 2 r\Delta f/c - f_\theta$.

{To determine the valid range for the frequency offsets that ensures an effective approximation with the first-order Taylor expansion Eq.(}\ref{taylor}{), we conduct numerical simulations under two scenarios: both carrier and receiving frequency offsets exist and only one frequency offset exists. The simulation results are shown in Fig.}\ref{Approximation_error0}{.				
			When only transmitting or frequency offset exists, the relative error of the approximation in Eq.(}\ref{taylor}{) is about $1 \%$ with $\sigma = 0.04 \Delta f$. the relative error increases to about $8 \%$ with $\sigma = 0.1 \Delta f$. When both transmitting and receiving frequency offsets exist, the relative error of the approximation in Eq.(}\ref{taylor}{) is about $1 \%$ with the standard deviation $\sigma = 0.02 \Delta f$ and increases to below  $10 \%$ for the standard deviation $\sigma$ up to $0.05 \Delta f$. When $\sigma = 0.1 \Delta f$, the relative error of the approximation exceeds $30 \%$, indicating a less effective approximation.}
			\begin{figure}[h] 
				\centering
				\includegraphics[width=0.45\textwidth]{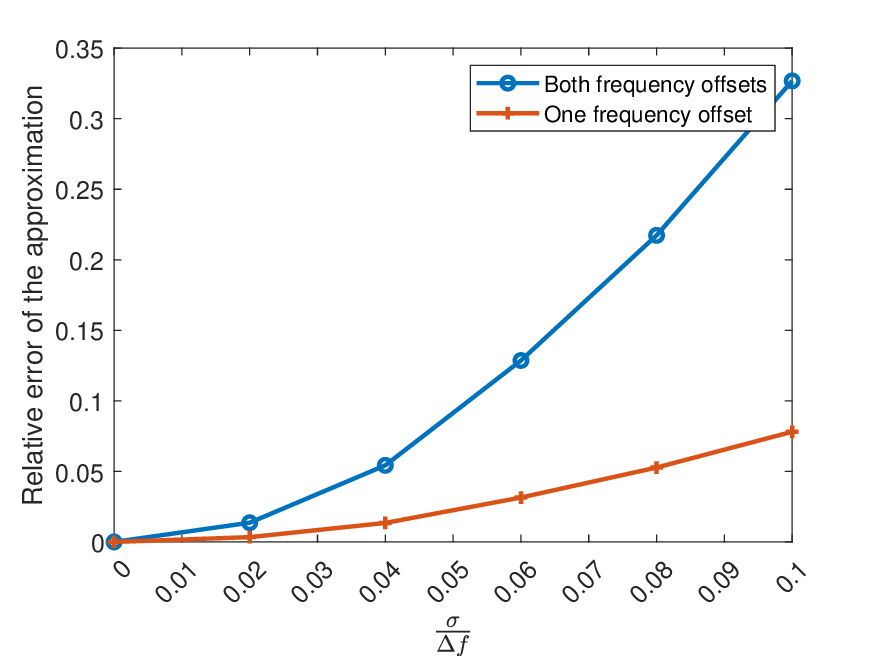}
				\caption{The relative error between the approximated signal matrix and the reference signal matrix.}
				\label{Approximation_error0}
			\end{figure}
			
			{In general, we consider the approximation effective and valid for $\sigma < 0.05 \Delta f$, and deem it ineffective when $\sigma > 0.1 \Delta f$. This boundary condition for the choice of Taylor expansion order, as supported by our numerical simulations, ensures the accuracy of our approximations under the specified conditions.}
   
The problem is to estimate the target range $r$ and the target angle $\theta$ in the model. In this paper, we analyze the noise characteristics caused by frequency offset, i.e., $\boldsymbol{N}_t$ and $\boldsymbol{N}_r$, together with their
influence on estimating range and DOA. According to the analysis, dedicated algorithms are applied estimate the target parameters.

\subsection{The Noise Characteristics}
As shown in Eq.(\ref{Ymn}), we have
\begin{equation}
\begin{split}
{Y}_{m,n} = & \beta e^{j 2 \pi [  (n-1) \Delta f \frac{2r}{c}- (m+n-2) f_\theta]}
\\ & + N_{t,m,n} + N_{r,m,n} + N_0
\end{split}
\end{equation}
{Two types of noises are considered in this paper: the Gaussian white noise and the equalized noise resulting from carrier and receiving frequency offsets. The Gaussian white noise represents the random, uncorrelated noise commonly encountered in signal processing. The noise resulting from carrier and receiving frequency offsets represent the interference on the signal matrix caused by the the carrier and receiving frequency offsets.}
If the transmitting and receiving frequency offsets follow i.i.d zero-mean Gaussian distribution among different pulses, the noise characteristics caused by the transmitting and receiving frequency offsets have the following properties.


\begin{theorem}\label{prop1}
{The noise caused by the receiving frequency offset will disturb the phase difference among both rows and columns in $\boldsymbol{Y}$. The noise caused by the transmitting frequency offsets will disturb the phase difference between different columns in $\boldsymbol{Y}$, but will not interfere with phase difference between different rows in $\boldsymbol{Y}$}. 
\end{theorem} 

\begin{proof}
According to the representation of $N_{t,m,n}$, i.e.,
\begin{equation}
\begin{split}
& N_{t,m,n} = 
 j 2 \pi \frac{\beta}{T_p} \sum_{i=1}^{N} f_{e,t,i}  e^{j 2 \pi [ (n-1) \Delta f \frac{2r}{c}- (m+i-2)f_\theta]} I_{1,i,n},
\end{split}
\end{equation}
The ratio between $N_{t,m,n_1}$ and $N_{t,m,n_2}$ $(n_1,n_2 \in [1,N]$, $n_1 \neq n_2)$ is
\begin{equation}
    \phi(n_1,n_2) = e^{j2\pi (n_1-n_2) \Delta f \frac{2 r}{c}} \frac{\sum_{i=1}^{N} f_{e,t,i} e^{-j 2 \pi \frac{i d \sin(\theta)}{c}} I_{1,i,n_1}}{\sum_{i=1}^{N} f_{e,t,i} e^{-j 2 \pi \frac{i d \sin(\theta)}{c}} I_{1,i,n_2}}.
\end{equation}

Since the transmitting frequency offsets of different antennas are independent among different pulses, phase differences between different columns in $\boldsymbol{N}_t$ are disturbed. Therefore, when using the phase differences between columns to estimate the target range, $\boldsymbol{N}_t$ will disturb the estimation, and $\boldsymbol{N}_t$ can only be treated as an additional noise when estimating the range. The ratio between $N_{t,m_1,n}$ and $N_{t,m_2,n}$ $(m_1,m_2 \in [1,M], m_1 \neq m_2)$ is
\begin{equation}
    \phi(m_1,m_2) = e^{-j 2 \pi (m_1 - m_2) f_\theta}.
\end{equation}
The phase differences between different rows in $\boldsymbol{N}_t$ remain the same as in steering vector $\boldsymbol{a}_r (\theta)$. {Therefore, the transmitting frequency offset does not influence the phase difference between different rows in $\boldsymbol{Y}$, and thus the estimation is not interrupted by transmitting frequency offsets when using the phase difference among the rows $\boldsymbol{Y}$ to estimate $\theta$.}

$N_{r,m,n}$ is represented as
\begin{equation}
\begin{split}
& N_{r,m,n} = j 2 \pi f_{e,r,m,n}\frac{2 r}{c} \beta  e^{j 2 \pi [ (n-1) \Delta f \frac{2r}{c}- (m+n-2) f_\theta]}
\\ & + j 2 \pi f_{e,r,m,n} \frac{\beta}{T_p} \sum_{i=1}^{N}  e^{j 2 \pi [ (n-1) \Delta f \frac{2r}{c}- (m+i-2) f_\theta]} I_{1,i,n}.
\end{split}
\end{equation}

The ratio between $N_{r,m,n_1}$ and $N_{r,m,n_2}$ $(n_1,n_2 \in [1,N], n_1 \neq n_2)$ is
\begin{equation}
\begin{split}
    & \phi(n_1,n_2) = 
\\ & \frac{f_{e,r,m,n_1} ( \frac{2 r}{c} T_p e^{-j 2 \pi n_1 f_\theta} + \sum_{i=1}^{N} e^{-j 2 \pi (m+i) f_\theta } I_{1,i,n_1})} 
 {f_{e,r,m,n_2} ( \frac{2 r}{c} T_p e^{-j 2 \pi n_2 f_\theta} + \sum_{i=1}^{N} e^{-j 2 \pi (m+i) f_\theta } I_{1,i,n_2})}.
\end{split}
\end{equation}
The receiving frequency offsets of different filter outputs at different antennas are also assumed to be independent among different pulses. Thus phase differences between different columns in $\boldsymbol{N}_r$ are disturbed. With a similar operation, it is easy to obtain that the phase differences between different rows in $\boldsymbol{N}_r$ are also disturbed. {Therefore, the receiving frequency offset will disturb the phase difference between different rows and different columns in $\boldsymbol{Y}$. Consequently, $\boldsymbol{N}_r$ is seen as an additional noise for both range and angle estimation.}
\end{proof}

{Different from the neglect of $\Delta f$ and $f_{e,r}$ within the array aperture in Eq.(}\ref{asu1}{), the impact of the frequency offsets on the estimation cannot be ignored. Eq.(}\ref{asu1}{) actually comes from the assumption $\Delta f \ll f_c$ and $f_{e,r} \ll f_c$, while the phase perturbation caused by the frequency offsets can be large enough to deteriorate the estimation accuracy. For example, the phase error caused by the receiving frequency offsets accumulates during the signal transmission, and the accumulated phase error concerns with the target range, which far exceeds the size of the the array aperture. Therefore, although the difference cause by the receiving frequency offset $f_{e,r}$ can be negligible within the array aperture in Eq.(}\ref{asu1}{), its effect on the estimation cannot be omitted.}

{In general, the transmitting frequency offsets will disturb the phase difference among vectors in $\boldsymbol{Y}$ but not the phase difference among rows because the transmitting frequency offset is a second-order small quantity compared with carrier frequency, but only a first-order small quantity compared with $\Delta f$. The receiving frequency offsets will disturb the phase difference among both rows and vectors in $\boldsymbol{Y}$ because the phase change caused by receiving frequency offsets is accumulated during the transmission process. The accumulated phase change cannot be neglected especially when the target range is large.}

\begin{theorem}\label{prop2}
The noise caused by the transmitting frequency offset added on the matched filter outputs is colored. 
\end{theorem}

\begin{proof}

 According to the assumption in~\cite{CFO_tr} and~\cite{CFO_r}, the transmitting and the receiving frequency offsets obey i.i.d zero-mean Gaussian distribution. We use $\sigma_t$ to represent the standard deviation of transmitting frequency offset and $\sigma_r$ to represent the standard deviation of receiving frequency offset, i.e.
\begin{equation}
f_{e,t,n} \sim \mathcal N ( 0,\sigma_t^2 ),\quad \forall n,
\end{equation}
\begin{equation}
f_{e,r,m,n} \sim \mathcal N ( 0,\sigma_r^2 ),\quad \forall m,n,
\end{equation}

Denote $\boldsymbol{n}_t \in \mathbb{C}^{MN \times 1},\boldsymbol{n}_r \in \mathbb{C}^{MN \times 1}$ as the vector form of $\boldsymbol{N}_t,\boldsymbol{N}_r$, respectively. The covariance matrix of noise $\boldsymbol{n}_{r}$ is denoted as $\boldsymbol{C}_r \in \mathbb{C}^{MN \times MN}$, and the covariance matrix of noise $\boldsymbol{n}_{t}$ is denoted as $\boldsymbol{C}_t \in \mathbb{C}^{MN \times MN} $. We use $C_{t,m,n,p,q}$ to represent the $[(n-1)M+m]$-th row and the $[(q-1)M+p]$-th column of $\boldsymbol{C}_t$, which denotes cross correlation between $Y_{m,n}$ and $Y_{p,q}$. $C_{r,m,n}$ denotes $[(n-1)M+m]$-th row and the $[(n-1)M+m]$-th column of $\boldsymbol{C}_r$. $C_{t,m,n,p,q}$ and $C_{r,m,n}$ are
\begin{equation}
\begin{split}
& C_{t,m,n,p,q} =
4 \pi^2 \| \alpha \|^2 \sigma_t^2 \frac{E}{N} \sum_{i=1}^{N}  e^{j (\phi_{n,q}^r - \phi_{m,p}^\theta)} I_{1,i,n} I_{1,i,q}^\text{*}.
\end{split}\label{eq1}
\end{equation}
\begin{equation}
\begin{split}
& C_{r,m,n} =   4 \pi^2 \| \alpha \|^2 \sigma_r^2 \frac{E}{N}  \bigg \{ T_p^2  \frac{4 r^2}{c^2} + \frac{2 r}{c} \sum_{i=1}^{N}  \big [ e^{-j \phi_{i,n}^\theta} I_{1,i,n}
\\ &  + e^{j \phi_{i,n}^\theta} I_{1,i,n}^\text{*} \big ] 
+ \sum_{i_1=1}^{N} \sum_{i_2=1}^{N} e^{-j \phi_{i_1,i_2}^\theta} I_{1,i_1,n} I_{1,i_2,n}^\text{*}  \bigg \},
\end{split}
\end{equation}
where
\begin{align}
    & \phi_{n,q}^r = 2 \pi (n-q) \Delta f \frac{2r}{c}, \\
    & \phi_{m,p}^\theta = 2 \pi (m-p)f_0 d \frac{\sin(\theta)}{c}.
\end{align}

The off-diagonal elements of $\boldsymbol{C}_r$ are $0$, and $\boldsymbol{C}_r$ is a diagonal matrix with different diagonal elements. $\boldsymbol{C}_t$ is not diagonal, and thus $\boldsymbol{N}_{t}$ is a colored noise. 
    
\end{proof}
Besides, from Eq.(\ref{eq1}), $C_{t,m,n,p,q}$ equals $C_{t,m',n,p',q}$ if $n = q$, $m-p = m' - p'$, which shows that $\boldsymbol{C}_t$ is a block Toeplitz matrix. Moreover, each block in $\boldsymbol{C}_t$ is a rank-$1$ matrix, which means that $\boldsymbol{C}_t$ is also a singular matrix.

{In physical mean, since the signal from a single transmitting antenna is received by all receiving antennas, the carrier frequency offset of that transmitting antenna influences the signals at all receiving points. This correlation leads to the noise being 'colored', thus impacting the signal processing and analysis.}


\section{ Signal Processing of Equalized Colored Noise}\label{sec4}




From the analysis in Section~\ref{sec2}, the noises $\boldsymbol{N}_r$ brought by the receiving frequency offsets are independent of the matched filter outputs. However, the noises $\boldsymbol{N}_t$ are correlated
at different matched filter outputs. Additional colored noise will deteriorate the accuracy of the target estimation~\cite{colored_noise,colored_noise2,colored_noise3}. Therefore, we need to mitigate the colored noise or whiten the colored noise according to the characteristics of colored noise. Since the colored noise matrix $\boldsymbol{C}_t$ is singular, it is inconvenient to use a transformation matrix to whiten $\boldsymbol{C}_t$. Therefore, we consider mitigating the colored noise with a fourth-order cumulant of the array data~\cite{four_level}. Moreover, we also consider sparse-signal denoising algorithms for the colored noise~\cite{IST,ANM,ANM2,ANM3,ANM4}.

\subsection{Whitening Colored Noise with the Fourth-Order Cumulant}

The $\gamma$-th order cumulant of $\boldsymbol{x} = [ x_1,x_2, \cdots, x_n]^\text{T}$ is defined as 
\begin{equation}
c_{\gamma_1, \cdots, \gamma_n} (\boldsymbol{x} )= \left.(-j)^\gamma \frac{\partial^\gamma \Psi\left(w_1, \cdots, w_n\right)}{\partial w_1^{\gamma_1} \cdots w_n^{\gamma_n}}\right|_{w_1=\ldots=w_n=0},
\end{equation}
where
\begin{align} 
\Psi \left(w_1, \cdots, w_n\right) &=  \ln \mathbb{E} \left\{e^{ j\left(w_1 x_1+w_2 x_2+\ldots+w_n x_n\right)}\right\},\\
\sum_{k=1}^{n} \gamma_k &= \gamma.
\end{align}
If $\gamma_k \in \{0,1\}, \forall k \in [1,n]$, we then use $c_\gamma$ to represent $c_{\gamma_1, \cdots, \gamma_n}$. According to~\cite{c4_1,c4_2}, ${c}_4 (x_i,x_j,x_p,x_q)$ can be represented as 
\begin{equation}
\begin{split}
& c_4 (x_i,x_j,x_p,x_q) =  \mathbb{E} \{ x_i x_j x_p^\text{*} x_q^\text{*} \} - \mathbb{E} \left  \{ x_i x_p^\text{*} \right \} \mathbb{E} \left  \{ x_j x_q^\text{*} \right \}
\\ &  - \mathbb{E} \{ x_i x_q^\text{*} \} \mathbb{E} \left \{ x_j x_p^\text{*}  \right \} - \mathbb{E} \{ x_i x_j \} \mathbb{E} \left \{ x_p^\text{*} x_q^\text{*}  \right \},
\end{split}
\end{equation}
where $i,j,p,q \in [1,MN] $ denote the $i,j,p,q$-th element of $\boldsymbol{x}$, respectively. $\mathbb{E} \{ \cdot \}$ denotes the expectation. The model is then transferred into a vector form as

\begin{equation}
\boldsymbol{y} = \beta \boldsymbol{a}(\theta,r) + \boldsymbol{n}_t + \boldsymbol{n}_r + \boldsymbol{n}_0,\label{eq2}
\end{equation}
where $\boldsymbol{y},\boldsymbol{n}_t,\boldsymbol{n}_r,\boldsymbol{n}_0$ are the vector form of $\boldsymbol{Y},\boldsymbol{N}_t,\boldsymbol{N}_r,\boldsymbol{N}_0$, respectively. $\boldsymbol{a}(\theta,r) \in \mathbb{C}^{MN \times 1}$ is defined as
\begin{equation}
    \boldsymbol{a}(\theta,r) = \boldsymbol{a}_t(\theta,r) \otimes \boldsymbol{a}_r(\theta).
\end{equation}

We can then obtain the fourth-order cumulant of $\boldsymbol{y}$ as follows
\begin{equation}
\begin{split}
c_4 (y_i,y_j,y_p,y_q)  = &
c_4 ( \beta a_i(\theta,r) + n_i , \beta a_j(\theta,r) + n_j  ,\\ & \beta a_p(\theta,r) + n_p , \beta a_q(\theta,r) + n_q ),
\end{split}
\end{equation}
where $y_i$ denotes the received signal at the $i$-th receiving antenna. $a_i(\theta,r)$ denotes the $i$-th entry of $\boldsymbol{a}(\theta,r)$. Since $i,j,p,q \in [1,MN]$, $c_4 (y_i,y_j,y_p,y_q)$ totally has $M^4 N^4$ kinds of value for the receiving antenna array. We then build a matrix $\boldsymbol{C}_4 \in \mathbb{C}^{M^2 N^2 \times M^2 N^2} $ with the $M^4 N^4$ values. The $[(i-1)MN+j]$-th row and the $[(p-1)MN+q]$-th column of $\boldsymbol{C}_4$ is
\begin{equation}
    \boldsymbol{C}_4 [ (i-1)MN + j, (p-1)MN + q] = c_4 (y_i,y_j,y_p,y_q),
\end{equation}



According to~\cite{c4_1}, $\boldsymbol{C}_4$ can be further represented as
\begin{equation}
    \boldsymbol{C}_4 = h (\boldsymbol{a}(\theta,r) \otimes \boldsymbol{a}^\text{*}(\theta,r)) (\boldsymbol{a}(\theta,r) \otimes \boldsymbol{a}^\text{*}(\theta,r))^\text{H},
\end{equation}
where $h$ is a constant concerning the target coefficient. Algorithms based on subspace decomposition can be applied on $\boldsymbol{C}_4$, and the angle and range of the target can be estimated with the eigenvectors. We first apply eigenvalue decomposition (EVD) on $\boldsymbol{C}_4$ to obtain the eigenvalues $\lambda_i,i=1,2,\cdots,MN$ and the corresponding eigenvectors $\boldsymbol{v}_i,i=1,2,\cdots,MN$. Assume that $\lambda_1 > \lambda_2 > \cdots > \lambda_{MN}$, and {the noise space $\boldsymbol{U}_n$ is denoted as}

\begin{equation} \label{subspace1}
    \boldsymbol{U}_n = \big[\boldsymbol{v}_1,\boldsymbol{v}_2,\cdots,\boldsymbol{v}_{MN} \big]. 
\end{equation}

{The spectrum concerning the target DOA and the target range can be then obtained as} 

\begin{equation} \label{subspace3}
    \boldsymbol{S}(\theta,r) = \frac{1}{(\boldsymbol{a}(\theta,r) \otimes \boldsymbol{a}(\theta,r)^\text{*}) \boldsymbol{U}_n \boldsymbol{U}_n^{\text{H}} (\boldsymbol{a}(\theta,r) \otimes \boldsymbol{a}(\theta,r)^\text{*})^{\text{H}}}.
\end{equation}

The target DOA and target range are then estimated with the coordinates of the peak of $\boldsymbol{S}(\theta,r)$.

\subsection{Denoising with Atomic Norm Minimization}
In this part, we consider sparse-signal denoising algorithms for the colored noise. According to (\ref{eq2}), the range and angle estimation problem can be written as a basis pursuit denoising problem as follows 
\begin{equation}
    \min_{x} \frac{1}{2} \left \| \boldsymbol{y} - \boldsymbol{A} \boldsymbol{x} \right \|_2^2 + \tau \left \| \boldsymbol{x} \right \|_1,\label{eq3}
\end{equation}
where $\boldsymbol{A}$ is a dictionary matrix defined as
\begin{equation}
\boldsymbol{A} = [\boldsymbol{a}(\theta_1,r_1), \boldsymbol{a}(\theta_2,r_2), \cdots, \boldsymbol{a}(\theta_{N_d},r_{N_d})]. 
\end{equation}
$N_d$ denotes the number of columns of dictionary matrix $\boldsymbol{A}$. $\theta_i$ and $\theta_j$ can be the same when $i \neq j$, which represents the same angle with different range, and likewise $r_i$ and $r_j$. $\tau$ is the regularization parameter and represents the strength of the sparsity constraint. $\boldsymbol{x}$ is a sparse vector. $x_k$, the $k$-th element of $\boldsymbol{x}$, equals $\beta$ if the $k$-th column of $\boldsymbol{A}$ represents the steering vector of true range and angle of the target and equals $0$ otherwise. 

An iterative soft thresholding (IST) algorithm is a typical algorithm for the problem~\cite{IST}. However, its performance is limited by grid mismatches. To improve estimation performance in CS problems, ANM is proposed \cite{ANM,ANM2,ANM3,ANM4}. Without constructing a dictionary matrix or discretizing spatial or range domain, ANM works as a gridless method and reconstructs sparse signals with convex optimization. For the problem in (\ref{eq3}), denoising ANM is proposed as an effective method. The atomic set is defined as 
\begin{equation}
\begin{split}
{\mathcal{A}} =  \left \{ \boldsymbol{a}(\theta,r), \forall \theta \in \left [ -\frac{\pi}{2},\frac{\pi}{2} \right ] , \forall r \in \left [ 0,r_{max} \right ]  \right \}.
\end{split}
\end{equation}

The atomic norm is then defined as
\begin{equation}
\begin{split}
\|\boldsymbol{x}\|_{{\mathcal{A}}} & \triangleq \inf \{ t>0:\boldsymbol{x} \in t  \mathrm{conv} ( \mathcal{A}) \}
\\ & = \inf \left\{\sum_k t_k \mid \boldsymbol{x}=\sum t_k {\boldsymbol{a}} (\theta_k, r_k ) \right \}.
\end{split}
\end{equation}

Denote $L$ as the total number of pulses. Collect the signal $\boldsymbol{y}$ of all $L$ pulses to form $\boldsymbol{X} \in \mathbb{C}^{MN \times L}$, which can be represented as

\begin{equation}
    \boldsymbol{X} = \boldsymbol{y} \times \boldsymbol{1}_{1 \times L} + \tilde{\boldsymbol{N}_t} + \tilde{\boldsymbol{N}_r} + \tilde{\boldsymbol{N}_0},
\end{equation}
where $\boldsymbol{1}_{1 \times L}$ denotes a row matrix with all the $L$ elements being $1$. Each column of $\tilde{\boldsymbol{N}_t}, \tilde{\boldsymbol{N}_r}, \tilde{\boldsymbol{N}_0}$ represent the additional noise $\boldsymbol{n}_t, \boldsymbol{n}_r, \boldsymbol{n}_0$ at each pulse.

According to Proposition~(1) in~\cite{ANM_2d}, an approximation of $\| \boldsymbol{x} \|_{\mathcal{A}}$ can be obtained as $\| \boldsymbol{x} \|_{\mathcal{T}}$ from a semi-definite program. For $\boldsymbol{X} \in \mathbb{C}^{MN \times L}$, $\| \boldsymbol{X} \|_{\mathcal{T}}$ is defined as 
\begin{equation}
    \| \boldsymbol{X} \|_{\mathcal{T}} = \min_{\boldsymbol{T},\boldsymbol{P}} \left \{ \frac{1}{2} \mathrm{Tr} (\mathrm{S}(\boldsymbol{T})) + \frac{1}{2} \mathrm{Tr} ( \boldsymbol{P} ) \mid 
\left[\begin{array}{cc}
\mathcal{S}(\boldsymbol{T}) & \boldsymbol{X} \\
\boldsymbol{X}^\text{*} & \boldsymbol{P}
\end{array}\right] \succeq 0 \right \},
\end{equation}
where $\mathrm{S} ( \boldsymbol{T} ) \in \mathbb{C}^{NM \times NM}$ is an two-fold Toeplitz matrix consisting of $N \times N$ block Toeplitz matrices, i.e.,
\begin{equation}
\mathcal{S}(\boldsymbol{T})=\left[\begin{array}{cccc}
\boldsymbol{T}_0 & \boldsymbol{T}_{-1} & \cdots & \boldsymbol{T}_{-\left(n_1-1\right)} \\
\boldsymbol{T}_1 & \boldsymbol{T}_0 & \cdots & \boldsymbol{T}_{-\left(n_1-2\right)} \\
\vdots & \vdots & \vdots & \vdots \\
\boldsymbol{T}_{n_1-1} & \boldsymbol{T}_{n_1-2} & \cdots & \boldsymbol{T}_0
\end{array}\right].
\end{equation}
Block $\boldsymbol{T}_l \in \mathbb{C}^{M \times M}$ is defined as 
\begin{equation}
\boldsymbol{T}_l=\left[\begin{array}{cccc}
x_{l, 0} & x_{l,-1} & \cdots & x_{l,-\left(n_2-1\right)} \\
x_{l, 1} & x_{l, 0} & \cdots & x_{l,-\left(n_2-2\right)} \\
\vdots & \vdots & \vdots & \vdots \\
x_{l, n_2-1} & x_{l, n_2-2} & \cdots & x_{l, 0}
\end{array}\right],
\end{equation}
where $\boldsymbol{x}_l$ is the $l$-th row of $\boldsymbol{T}$, and $x_{l,p}$ represents the $p$-th element of $\boldsymbol{x}_l$. 

Therefore, with the above transformation, the denoising problem can be reformulated as
\begin{equation}
\begin{aligned}\min_{\boldsymbol{\hat{X}},\boldsymbol{T},\boldsymbol{P}} &  \frac{1}{2} \mathrm{Tr} (\mathrm{S}(\boldsymbol{T})) + \frac{1}{2} \mathrm{Tr} ( \boldsymbol{P} ) \\
\text { s.t. } &\left[\begin{array}{cc}
\mathcal{S}(\boldsymbol{T}) & \boldsymbol{\hat{X}} \\
\boldsymbol{\hat{X}}^\text{*} & \boldsymbol{P}
\end{array}\right] \succeq 0 ,
\\ & \quad \| \boldsymbol{X} - \boldsymbol{\hat{X}} \|_F^2 \leq \tau.
\end{aligned}
\end{equation}
$\tau$ is the regularization parameter, which determines the strength of the sparsity constraint. $\boldsymbol{\hat{X}}$ is the denoised signal matrix. {Target estimation can be achieved by employing subspace methods on $\boldsymbol{X}$. We first use EVD on $\hat{\boldsymbol{X}}$ to obtain the eigenvalues
			 $\lambda_i,i=1,2,\cdots,MN$ and the corresponding eigenvectors $\boldsymbol{v}_i,i=1,2,\cdots,MN$. Assume that $\lambda_1 > \lambda_2 > \cdots > \lambda_{MN}$, and the noise space $\boldsymbol{U}_n$ is denoted as}
			
			\begin{equation}
				\boldsymbol{U}_n = \big[\boldsymbol{v}_1,\boldsymbol{v}_2,\cdots,\boldsymbol{v}_{MN} \big]. 
			\end{equation}
			
			{The spectrum concerning the target DOA and the target range can be then obtained as }
			
			\begin{equation} \label{subspace2}
				\boldsymbol{S}(\theta,r) = \frac{1}{\boldsymbol{a}(\theta,r) \boldsymbol{U}_n \boldsymbol{U}_n^{\text{H}} \boldsymbol{a}^{\text{H}}(\theta,r)}.
			\end{equation}
			
			{The target DOA and target range correspond to the coordinates of the peak of $\boldsymbol{S}(\theta,r)$.}

\section{The CRLB for Range-Angle Estimation}\label{sec3}
For the model (\ref{eq16}), we derive the CRLB for the estimation of both range and angle in this section. According to the vector form model (\ref{eq2}), $\boldsymbol{y}$ obeys Gaussian distribution, i.e.,
\begin{equation}
    \boldsymbol{y} \sim \mathcal 
    {CN} ( \beta \boldsymbol{a}(\theta,r), \boldsymbol{C} ),
\end{equation}
where 
\begin{equation}
\begin{split}
    \boldsymbol{C} \triangleq & \mathbb {E} \left \{ (\boldsymbol{y} - \beta \boldsymbol{a}(\theta,r) ) (\boldsymbol{y} - \beta \boldsymbol{a}(\theta,r) ) ^\text{H} \right \}
\\ = & \boldsymbol{C}_0 +\boldsymbol{C}_t + \boldsymbol{C}_r,
\end{split}
\end{equation}
$\boldsymbol{C}_0\triangleq \mathbb{E} \left \{ \boldsymbol{n}_0 \boldsymbol{n}_0^\text{H} \right \}$, $\boldsymbol{C}_t \triangleq \mathbb{E} \left \{ \boldsymbol{n}_t \boldsymbol{n}_t^\text{H} \right \} \triangleq \mathbb{E} \left \{ \boldsymbol{n}_r \boldsymbol{n}_r^\text{H} \right \}$ denote the covariance matrices of $\boldsymbol{n}_0$, $\boldsymbol{n}_t$ and $\boldsymbol{n}_r$, respectively. Assume that $\boldsymbol{n}_0 \sim \mathcal {CN} ( 0,\sigma_0^2 )$, we have $\boldsymbol{C}_0 = \sigma_0^2 \boldsymbol{I}$. $\boldsymbol{C}_r$ is a diagonal matrix, and the $((n-1)M+m)$-th row, $((n-1)M+m)$-th column of $\boldsymbol{C}_r$ is
\begin{equation}
\begin{split}
 & \boldsymbol{C}_{r,(n-1)M+m,(n-1)*M+m} =   4 \pi^2 \| \beta \|^2 \sigma_r^2   \bigg \{ T_p^2  \frac{4 r^2}{c^2} +
 \\ & \qquad \qquad  \frac{4 r}{c} \sum_{i=1}^{N} \mathcal{R} \{ e^{-j 2 \pi (i-n) f_\theta} I_{1,i,n} \} + 
\\   & \qquad \qquad  \sum_{i_1=1}^{N} \sum_{i_2=1}^{N} e^{-j 2 \pi (i_1 - i_2) f_\theta} I_{1,i_1,n} I_{1,i_2,n}^\text{*}  \bigg \}.
\end{split}
\end{equation}
The $((n-1)M+m)$-th row, $((b-1)M+a)$-th column of $\boldsymbol{C}_t$ is 
\begin{equation}
\begin{split}
 &  \boldsymbol{C}_{t,(n-1)M+m,(b-1)M+a}
 \\ &  \quad = 4 \pi^2 \sigma_t^2 \sum_{i=1}^{N} e^{j 2 \pi [ (n-a) \Delta f \frac{2 r}{c} - (m-a) f_\theta ]} I_{1,i,n} I_{1,i,b}^\text{*},
\end{split}
\end{equation}

With the obtained covariance matrix, the probability
density function of Gaussian distribution $ \boldsymbol{y} \sim \mathcal CN ( \beta \boldsymbol{a}(\theta,r), \boldsymbol{C} )$ can be expressed as
\begin{equation}
\begin{split}
& f(\boldsymbol{y}) = \frac{1}{\pi^{MN} \| C \|^{1/2}} 
 e^{  - \frac{1}{2} (\boldsymbol{y} - \beta \boldsymbol{a}(\theta,r))^\text{H} \boldsymbol{C}^{-1} (\boldsymbol{y} - \beta \boldsymbol{a}(\theta,r))},
\end{split}
\end{equation}
where $\| \cdot \|$ denotes the determinant.

To derive the CRLB, Fisher information matrix (FIM) must be first obtained, where the second derivative of $\mathrm{ln} p(\boldsymbol{y} ; \boldsymbol{s})$ to all variables in $\boldsymbol{s}$ are needed. $\boldsymbol{s}$ is defined as
\begin{equation}
\boldsymbol{s} = \left[ \boldsymbol{\theta}^\text{T} \quad \boldsymbol{r}^\text{T} \quad \boldsymbol{\beta}^\text{T}  \right]^\text{T}.
\end{equation}
We denote $\boldsymbol{F} (s_i,s_j)$ as the second derivative of $\mathrm{ln} p(\boldsymbol{y} ; \boldsymbol{s})$ to $s_i$ and $s_j$, where $s_i$ and $s_j$ are variables in $\boldsymbol{s}$. $\boldsymbol{F} (s_i,s_j)$ is
\begin{equation}
\begin{split}
& \boldsymbol{F}(s_i,s_j)= -\mathbb{E} \{ \frac{\partial^2 \mathrm{ln} f(\boldsymbol{y} ; \boldsymbol{s})}{\partial s_i \partial s_j}\} .
\end{split}
\end{equation}

After calculation and simplification, we obtain $\boldsymbol{F}(r,r)$ as follows

\begin{equation}
\begin{split}
 \boldsymbol{F}(r,r) & = 
\bigg \{ \mathcal{R} \left \{ \|\beta\|^2 \frac{\partial \boldsymbol{a}(\theta,r)}{\partial r} \boldsymbol{C}^{-1} \frac{\partial \boldsymbol{a}(\theta,r)}{\partial r} \right \} 
\\ & \qquad \qquad + \frac{1}{2} \mathrm{Tr} \{ \boldsymbol{C}^{-1} \frac{\partial \boldsymbol{C}}{\partial r} \boldsymbol{C}^{-1} \frac{\partial \boldsymbol{C}}{\partial r} \} \bigg \}^{-1}.
\end{split}
\end{equation}

As claimed before, since $\boldsymbol{N}_t$ does not interfere with the estimation of angle, the covariance matrix of $\boldsymbol{C}_t$ should not be included in $\boldsymbol{C}$ when deriving the CRLB of angle, i.e.,
\begin{equation}
\begin{split}
\boldsymbol{F}(\theta,\theta) & =  \mathcal{R} \left \{ \|\beta\|^2 \frac{\partial \boldsymbol{a}(\theta,r)}{\partial \theta} \boldsymbol{\tilde{C}}^{-1} \frac{\partial \boldsymbol{a}(\theta,r)}{\partial \theta} \right \} 
\\ & + \frac{1}{2} \mathrm{Tr} \{ \boldsymbol{\tilde{C}}^{-1} \frac{\partial \boldsymbol{\tilde{C}}}{\partial \theta} \boldsymbol{\tilde{C}}^{-1} \frac{\partial \boldsymbol{\tilde{C}}}{\partial \theta} \},
\end{split}
\end{equation}
where
\begin{equation}
\boldsymbol{\tilde{C}} = \boldsymbol{C}_0 + \boldsymbol{C}_r.
\end{equation}

 Finally, the CRLB with respect to $\boldsymbol{s}_i$ is 
\begin{equation} 
\mathrm{CRLB} \{\boldsymbol{s}_i \} = \left[ \boldsymbol{F}^{-1} \right]_{i,i} , 
\end{equation} 
where $ \boldsymbol{F}^{-1} $ denotes the inverse of $\boldsymbol{F}$ and $ \left[ \cdot \right ] _{i,i}$ denotes the element at the $i$-th column and the $i$-th row of the matrix. The derivation of CRLB using $\boldsymbol{F}^{-1}$ requires that $\boldsymbol{F}$ is positive defined~\cite{FIM_singular}. However, since the value of entries corresponding to the lower bound of range in FIM greatly differs from those corresponding to the lower bound of angle, FIM can be singular, especially when the dimension of variables is high. Therefore, we use a lower bound of FIM to describe the performance of the estimation as~\cite{FIM_2}. 
\begin{equation}
    \mathrm{CRLB}\{s_i \} \geq {\boldsymbol{F}}_{i,i}^{-1}.
\end{equation}
{The derived CRLB simultaneously considers the influence of additional white noise, carrier frequency offsets and receiving frequency offsets. Moreover, the equalized noise caused by the frequency offsets can be colored. Unlike white noise, which has a diagonal covariance matrix due to its uncorrelated nature, colored noise is characterized by a covariance matrix with non-zero off-diagonal elements. The derived CRLB applies to the scenarios with colored noise caused by frequency offsets.}

\section{Simulation Results and Analysis}\label{sec5}

\begin{table}[h]
\renewcommand{\arraystretch}{1.2}
\caption{Simulation Parameters.}\label{table1}
\begin{tabular}{cc}
\hline
\multicolumn{1}{c}
{\textbf{Parameter}}                       & \multicolumn{1}{c}{Value} \\ \hline
The signal-to-noise ratio (SNR) of received signal             & $20$~dB                     \\
The number of transmitting antennas N                          & $4$                          \\
The number of receiving antennas M                             & $4$                          \\
The number of targets S                                        & $1$                          \\
The space between antennas d                                   & $0.5$ wavelengths             \\
The number of pulses $L$                        & $200$                       \\
The carrier frequency $f_0$                                     & $10$~G Hz                    \\
The default frequency difference $ \Delta f$               & $10$~k Hz                    \\
The default range of target $ r$               & $6000$~m                    \\
The default DOA of target $ \theta$               & $30$~°                    \\
The default standard deviation of transmitting   & \multirow{2}{*}{$500$~Hz} \\ 
frequency offset $\sigma_t$ &                     \\
The default standard deviation of receiving  & \multirow{2}{*}{$500$~Hz} \\ 
 frequency offset $\sigma_r$ &       \\ \hline 
\end{tabular}
\end{table}

Simulation parameters are given in Table~\ref{table1}. The number of Monte Carlo simulations is $ 10^3$. Attention that SNR in Table.~\ref{table1} refers to the power ratio between the signal and the white noise, where the noise caused by the frequency offset is not considered. Assume that the transmitting and receiving frequency offsets remain unchanged in one pulse duration and obey i.i.d Gaussian distribution among different pulses. We use $\sigma_t$ to represent the standard deviation of the transmitting frequency offset and $\sigma_r$ to represent the standard deviation of the receiving frequency offset, i.e.,
\begin{equation}
f_{e,t,n} \sim \mathcal N ( 0,\sigma_t^2 ),\quad \forall n,
\end{equation}
\begin{equation}
f_{e,r,m,n} \sim \mathcal N ( 0,\sigma_r^2 ),\quad \forall m,n,
\end{equation}

We first show the deterioration caused by frequency offsets. According to the model in (\ref{eq0}), the matched filter outputs are decided by $e^{j 2 \pi f_{e,r,m,n}\frac{2r}{c}}$ and $\int_{0}^{T_p} e^{-j 2 \pi (f_{e,t,i}-f_{e,r,m,n})t} dt$. When pulse duration is set to distinguish the signal of different frequencies exactly, i.e., $T_p = 1/\Delta f$, the integration $\int_{0}^{T_p} e^{-j 2 \pi (f_{e,t,i}-f_{e,r,m,n})t} dt$ is fixed when the ratio ${(f_{e,t,i}-f_{e,r,m,n})}/{\Delta f}$ is fixed. If the ratio ${r}/{r_{\max}}$ and ${f_{e,r,m,n}}/{\Delta f}$ are both fixed, $e^{j 2 \pi f_{e,r,m,n} \frac{2r}{c}}$ is fixed as well. In general, when the pulse duration of FDA-MIMO radar equals the reciprocal of $\Delta f$ and the ratio of targets' range to maximum detection range is unchanged, deterioration caused by frequency offsets is only decided by the ratio of the standard deviation of frequency offsets and the frequency difference, i.e., ${\sigma_t}/{\Delta f}$ and ${\sigma_r}/{\Delta f}$.

For example, in scene $1$, $\Delta f$ is set as $10$~kHz. The maximum unambiguous detection range is then $15$~km. The pulse duration is set as the minimum length needed to distinguish different frequencies, i.e., $0.1$~ms. In scene $2$, $\Delta f$ is set as $1$~kHz. The maximum detection range is $150$~km and the pulse duration is $1$~ms. When detecting a target at $6$~km away in scene $1$ with a standard deviation of transmitting frequency offset $1k$~Hz, the deterioration of frequency offsets equals that when detecting a target at $60$~km away in scene $2$ with a standard deviation of transmitting frequency offset $100$~Hz. In  Fig.~\ref{fig2}, real noise represents the difference between the matched filter outputs and the theoretical outputs, i.e., $\left[ \boldsymbol{y} - \beta \boldsymbol{a}(\theta,r) \right]$. Estimated noise represents the noise generated according to  (\ref{eq4}) and (\ref{eq5}). We use equalized SNR, i.e., the power ratio between signal and noise caused by frequency offsets, to describe the deterioration of frequency offsets. The equalized SNR under both scenes is simulated and shown in Fig.~\ref{fig2} {, and numerical values of the simulation result is shown in Table.}\ref{table2}. Obviously, the equalized SNR is the same with the same ${\sigma_t}/{\Delta f}$ in both scenes, which is consistent with the analysis above. 

\begin{figure}[h]
\centering
\includegraphics[width=3.6 in]{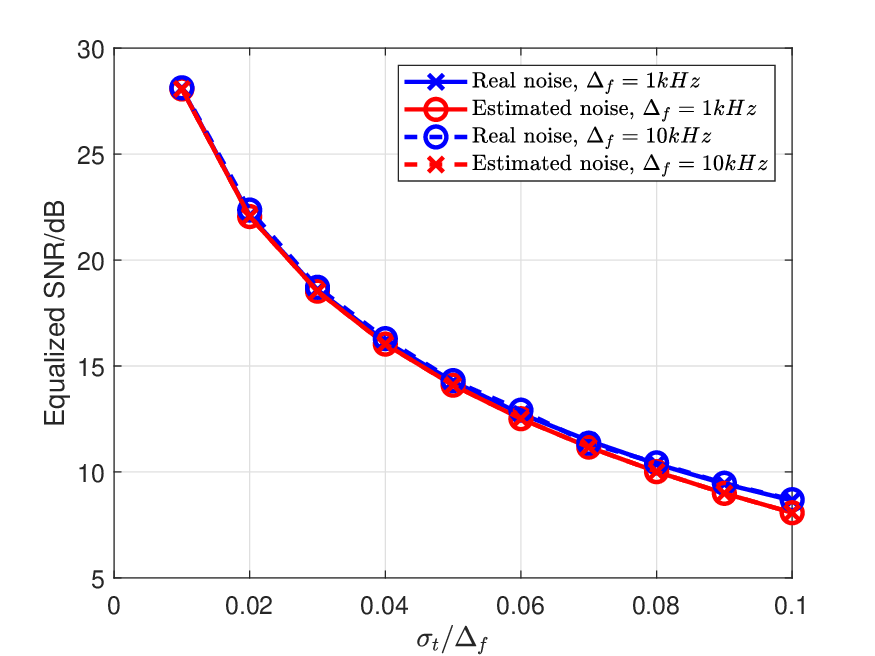}
\caption{Equalized SNR with a different standard deviation of transmitting frequency offset when $\Delta f = 1$~kHz and $10$~kHz.}
\label{fig2}
\end{figure}

\begin{table}[h]
\caption{Equalized SNR with different standard deviation of transmitting frequency offset.}\label{table2}
\centering
\renewcommand{\arraystretch}{1.4}
\begin{tabular}{|cll|cccc|}
\hline
\multicolumn{3}{|c|}{\multirow{4}{*}{$\frac{\sigma}{\Delta f}$}} & \multicolumn{4}{c|}{\multirow{2}{*}{\begin{tabular}[c]{@{}c@{}}Equalized SNR with transmitting frequency offset\end{tabular}}} \\
\multicolumn{3}{|c|}{}                         & \multicolumn{4}{c|}{}                                                                                                             \\ \cline{4-7} 
\multicolumn{3}{|c|}{}                         & \multicolumn{2}{c|}{$\Delta f=1kHz$}                                          & \multicolumn{2}{c|}{$\Delta f= 10kHz$}                  \\ \cline{4-7} 
\multicolumn{3}{|c|}{}                         & \multicolumn{1}{c|}{Estimation}     & \multicolumn{1}{c|}{Actual value}    & \multicolumn{1}{c|}{Estimation}    & Actual value    \\ \hline
\multicolumn{3}{|c|}{0.02}                     & \multicolumn{1}{c|}{22.05}          & \multicolumn{1}{c|}{22.09}           & \multicolumn{1}{c|}{22.05}         & 18.72           \\ \hline
\multicolumn{3}{|c|}{0.04}                     & \multicolumn{1}{c|}{16.03}          & \multicolumn{1}{c|}{16.15}           & \multicolumn{1}{c|}{16.03}         & 14.31           \\ \hline
\multicolumn{3}{|c|}{0.06}                     & \multicolumn{1}{c|}{12.51}          & \multicolumn{1}{c|}{12.74}           & \multicolumn{1}{c|}{12.51}         & 12.92           \\ \hline
\multicolumn{3}{|c|}{0.08}                     & \multicolumn{1}{c|}{10.01}          & \multicolumn{1}{c|}{10.38}           & \multicolumn{1}{c|}{10.01}         & 10.43           \\ \hline
\multicolumn{3}{|c|}{0.1}                      & \multicolumn{1}{c|}{8.07}           & \multicolumn{1}{c|}{8.65}            & \multicolumn{1}{c|}{8.07}          & 8.69            \\ \hline
\end{tabular}
\end{table}


When ${r}/{r_{max}} = 0.4$, ${\sigma_t}/{\Delta f} = 0.05$, the equalized SNR is $14.3 $~dB. When ${\sigma_t}/{\Delta f}$ rises to $0.1$, the equalized SNR is $8.7 $~dB. From the comparison between the estimated noise and the real noise, it is obvious that the power of noise is accurately estimated, which proves the effectiveness of the model.

We then show the relationship between the equalized SNR and the receiving frequency offset. Similar to the analysis above, when $T_p = {1}/{\Delta f}$, the equalized SNR only concerns with ${r}/{r_{max}}$ and ${\sigma_r}/{\Delta f}$. Let ${r}/{r_{max}} = 0.4$, the relation between the equalized SNR and the standard deviation of receiving frequency offset $\sigma_r$ is shown in Fig.~\ref{fig4} {, and numerical values of the simulation result is shown in Table.}\ref{table3}.

\begin{figure}[h]
\centering
\includegraphics[width=3.6 in]{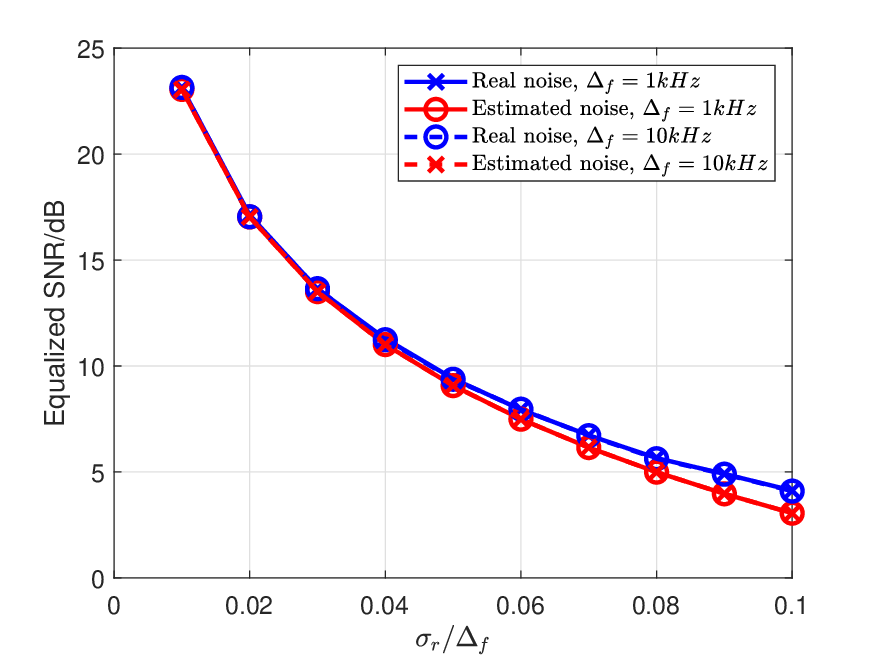}
\caption{Equalized SNR with a different standard deviation of receiving frequency offset when $\Delta f = 1$~kHz and $10$~kHz.}
\label{fig4}
\end{figure}

\begin{table}[h]
\caption{Equalized SNR with different standard deviation of receiving frequency offset.}\label{table3}
\centering
\renewcommand{\arraystretch}{1.4}
\begin{tabular}{|cll|cccc|}
\hline
\multicolumn{3}{|c|}{\multirow{4}{*}{$\frac{\sigma}{\Delta f}$}} & \multicolumn{4}{c|}{\multirow{2}{*}{\begin{tabular}[c]{@{}c@{}}Equalized SNR with receiving frequency offset\end{tabular}}} \\
\multicolumn{3}{|c|}{}                         & \multicolumn{4}{c|}{}                                                                                                          \\ \cline{4-7} 
\multicolumn{3}{|c|}{}                         & \multicolumn{2}{c|}{$\Delta f=1kHz$}                                         & \multicolumn{2}{c|}{$\Delta f= 10kHz$}                \\ \cline{4-7} 
\multicolumn{3}{|c|}{}                         & \multicolumn{1}{c|}{Estimation}    & \multicolumn{1}{c|}{Actual value}    & \multicolumn{1}{c|}{Estimation}   & Actual value   \\ \hline
\multicolumn{3}{|c|}{0.02}                     & \multicolumn{1}{c|}{17.03}         & \multicolumn{1}{c|}{17.12}           & \multicolumn{1}{c|}{17.03}        & 17.04          \\ \hline
\multicolumn{3}{|c|}{0.04}                     & \multicolumn{1}{c|}{11.01}         & \multicolumn{1}{c|}{11.29}           & \multicolumn{1}{c|}{11.01}        & 11.24          \\ \hline
\multicolumn{3}{|c|}{0.06}                     & \multicolumn{1}{c|}{7.49}          & \multicolumn{1}{c|}{7.94}            & \multicolumn{1}{c|}{7.49}         & 7.96           \\ \hline
\multicolumn{3}{|c|}{0.08}                     & \multicolumn{1}{c|}{4.99}          & \multicolumn{1}{c|}{5.65}            & \multicolumn{1}{c|}{4.99}         & 5.62           \\ \hline
\multicolumn{3}{|c|}{0.1}                      & \multicolumn{1}{c|}{3.05}          & \multicolumn{1}{c|}{4.12}            & \multicolumn{1}{c|}{3.05}         & 4.09           \\ \hline
\end{tabular}
\end{table}

From the comparison between Fig.~\ref{fig2} and Fig.~\ref{fig4}, when ${r}/{r_{max}} = 0.4$, the power of noise caused by receiving frequency offset is about $5$~dB higher than the power of noise caused by transmitting frequency offset with the same standard deviation.

We then show the equalized SNR with different ranges of targets. From the representation of covariance matrix $\boldsymbol{C}_t$ and $\boldsymbol{C}_r$, the power of noise brought from the receiving frequency offset concerned with the range of the target, while the power of noise brought from transmitting frequency offset does not. We simulate the equalized SNR with only transmitting frequency offset and only receiving frequency offset, respectively. The simulation results are shown in Fig.~\ref{fig5}.

\begin{figure}[h]
\centering
\includegraphics[width=3.6 in]{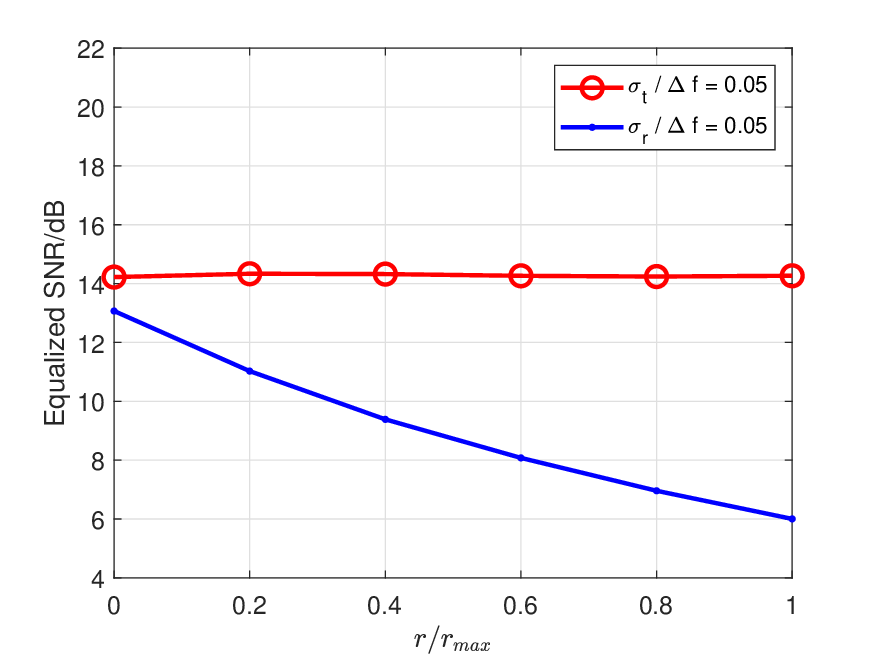}
\caption{equalized SNR with different ranges of target.}
\label{fig5}
\end{figure}

As shown in Fig.~\ref{fig5}, when only transmitting frequency offset exists, the equalized SNR is not influenced by the range of targets. With only receiving frequency offset, noise to signal ratio increases when the detected target is further. As the simulation result shows, when ${\sigma_r}/{\Delta f} = 0.05$, the difference between the maximum equalized SNR and the minimum equalized SNR can exceed $7 $~dB with different target ranges.

{Before comparing the estimation performances of different algorithms, we first show the used algorithms in the simulation. To the best of our knowledge, there are currently no studies working on analyzing the influence of the frequency offsets on FDA-MIMO system under the scenario where the frequency offsets are assumed to obey a certain distribution in different pulses. Therefore, for comparison, we adopt MUSIC algorithm, OMP algorithm, sparse signal denoising algorithm (atomic norm minimization) and the algorithms based on the fourth-order cumulant to address the parameter estimation problem. The latter two algorithms, in particular, are employed to try to mitigate the equalized colored noise caused by frequency offsets. The used algorithms are listed as follows}
\begin{itemize}
        \setlength{\parsep}{-10pt}
	\item {2D MUSIC:  Use 2D MUSIC algorithm to co-estimate the target range and DOA with the signal matrix $\boldsymbol{Y}$.}
	\item {2D MUSIC-C:  Use 2D MUSIC algorithm to co-estimate the target range and DOA with the fourth-order cumulant matrix $\boldsymbol{C}_4$.}
	\item {MUSIC-R: Use MUSIC algorithm to estimate the target DOA only with the phase differences between different rows in the signal matrix $\boldsymbol{Y}$.}
	\item {OMP:  Use OMP algorithm to co-estimate the target range and DOA with the signal matrix $\boldsymbol{Y}$.}
	\item {OMP-ANM: Use 2fold ANM algorithm to denoise the signal matrix $\boldsymbol{Y}$ first, and use OMP algorithm to co-estimate the target range and DOA with the denoised signal matrix.}
\end{itemize}
{The computational complexity of the proposed algorithms is shown in Table} \ref{Complexity} {, with $N_d$ representing the number of columns in the dictionary.}
	\begin{table}[H] \caption{Computational complexity}
 \centering
		\begin{tabular}{|c|c|} 
			\hline
			Algorithms                        & Computational Complexity                                     \\ \hline
			2D MUSIC           & \makecell[c]{$ 
\mathcal{O} \big( M^3 N^3 \big) $ } \\ \hline
			2D MUSIC-C       & \makecell[c]{$ \mathcal{O} \big( M^6 N^6 \big) $  }            \\ \hline
                MUSIC-R       & \makecell[c]{$ \mathcal{O} \big( M^3 \big) $  }            \\ \hline
			\makecell[c]{OMP} & $\mathcal{O} \big( MNLN_d \big)  $                                      \\ \hline
			OMP-ANM                               & $\mathcal{O} \big( (MN+L)^{3.5} + MNLN_d \big)$                           \\ \hline  
		\end{tabular}
	\label{Complexity}
	\end{table}
{The practical computational time of the algorithms are shown in Table.} \ref{Time} {. All the simulation results are obtained on a PC with Matlab R2021b with a 3.6 GHz Intel Core i7.}
\begin{table}[h] \caption{Computational Time}
\centering
\begin{tabular}{|c|cll|}
\hline
\multicolumn{1}{|l|}{} & \multicolumn{3}{l|}{Computational Time / s} \\ \hline
2D MUSIC             & \multicolumn{3}{c|}{0.040}                   \\ \hline
2D MUSIC-C         & \multicolumn{3}{c|}{6.155}                   \\ \hline
MUSIC-R   & \multicolumn{3}{c|}{0.011}                   \\ \hline
OMP                    & \multicolumn{3}{c|}{0.851}                    \\ \hline
OMP-ANM                    & \multicolumn{3}{c|}{26.30}                    \\ \hline
\end{tabular}
\label{Time}
\end{table}

{We then show the influence of the frequency offsets on estimating the range and angle of targets in Fig.}~\ref{fig6}\subref{fig6a} {and Fig.}~\ref{fig6}\subref{fig6b}{, respectively. We assume only one target and set the target range as ${r} =0.4r_{max}$ and the target DOA as $\theta = 30^{\circ}$. To emphasize the influence of frequency offsets and avoid the influence of additive white noise, the SNR is set as $50$~dB. Fig.}~\ref{fig6}\subref{fig6a} {shows the relation between the root mean square error (RMSE) of estimated range and the standard deviation of carrier frequency offset. In scenarios with a single target, OMP algorithm outperforms MUSIC algorithm across both low and high SNR ranges. However, in scenarios with two targets, the performance of OMP algorithm significantly deteriorates compared with MUSIC algorithm, to the extent that it is outperformed by MUSIC algorithm. The results indicate the sensitivity of OMP algorithm to the number of targets.}

{In single-target scenarios, the interference arises from both additive white noise and the equalized colored noise caused by frequency offsets. Colored noise affects the off-grid elements of the covariance matrix, impacting the eigenvector decomposition used by MUSIC algorithm. Hence, MUSIC is more sensitive to noise, particularly to the colored noise introduced by frequency offsets. As the number of targets increases, the mutual correlation among signals reflected by different targets becomes a significant factor. Considering that OMP algorithm uses the correlation between the signals to determine the target parameters, the additional complexity and interference introduced by multiple targets can lead to a significant challenge for OMP algorithm. Consequently, in scenarios with two or more targets, the OMP algorithm struggles to maintain the level of accuracy and robustness exhibited in single-target situations.}

\begin{figure}[h]
\subfloat[Range estimation]{
\label{fig6a}
\includegraphics[width=3.6in]{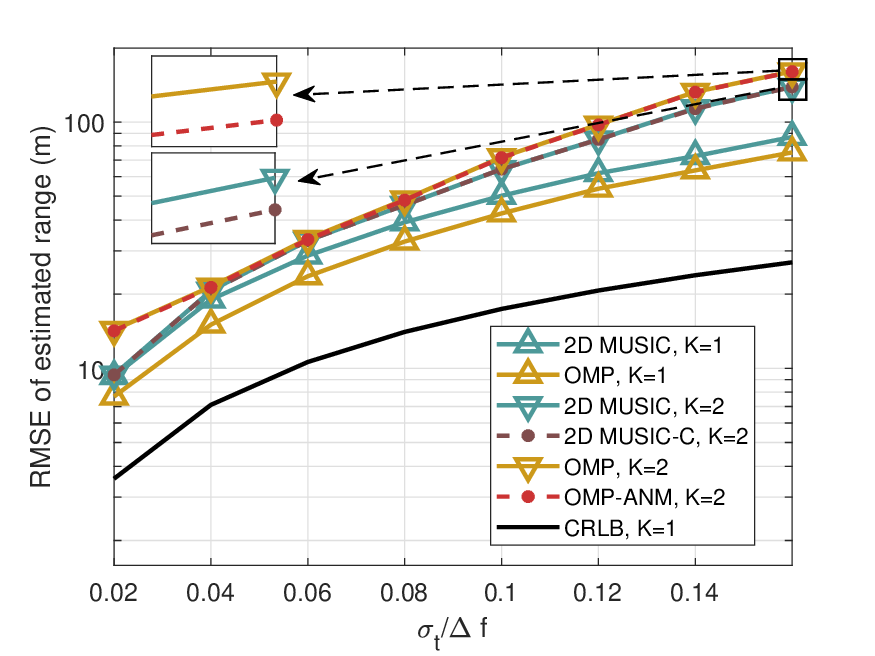}}
\quad
\subfloat[DOA estimation]{
\label{fig6b}  
\includegraphics[width=3.6 in]{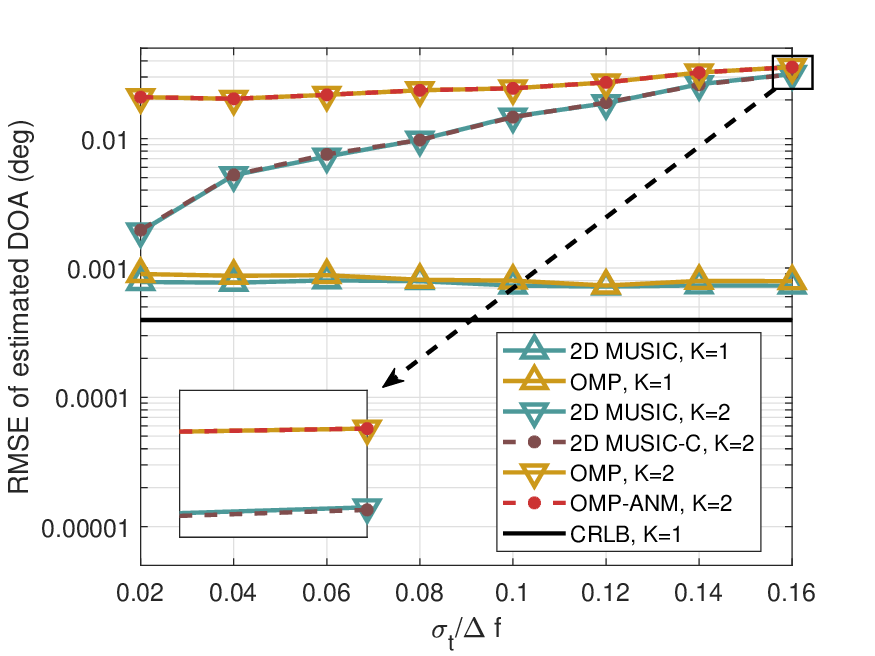}}
\caption{Influence of transmitting frequency offset on estimation.}
\label{fig6}
\end{figure}

{Compared with the distinct differences between the performances of OMP algorithm and MUSIC algorithm, the differences among algorithms within the same base category (MUSIC or OMP) are minimal. According to the simulation results in both Fig.}~\ref{fig6} \subref{fig6a} {and Fig.}~\ref{fig6} \subref{fig6b}{, the 2D-MUSIC algorithm using the fourth-order cumulant shows little improvement over the standard 2D-MUSIC algorithm. Similarly, the OMP algorithm based on atomic norm minimization offers marginal improvement compared to the conventional OMP algorithm. These results indicate that the denoising techniques we adopted show limited superiority over the common algorithms in the presence of frequency offsets.
}

{We then analyze the simulation results in Fig.}~\ref{fig6} \subref{fig6b}{, which shows the relation between the RMSE of the estimated target DOA and the carrier frequency offset. In single target scenario, the performances of the used algorithms remain basically unchanged. When the number of target equals $2$, the DOA estimation deteriorate with the increase of the transmitting frequency offset. The different simulation results in single and dual target scenarios comes from the the requirement of more phase information when estimating multiple targets. Both 2D MUSIC algorithm and OMP algorithm estimate the target range and DOA jointly with the phase differences among both the rows and columns of the matrix $\boldsymbol{Y}$. According to Proposition} \ref{prop1}{, the phase difference between rows is solely related to the target DOA and not influenced by carrier frequency offset, which indicates that the phase difference between rows is the primary influencing factor in the estimation of DOA. Therefore, the DOA estimation accuracy remains consistent even in the presence of significant carrier frequency offset. However, in scenarios with multiple targets, more phase information are required to distinguish different targets, which indicates that the phase differences between both rows and columns both have significant influence on the DOA estimation. As described in Proposition} \ref{prop1}{, the phase difference between columns concern with the transmitting frequency offset. Consequently, the DOA estimation performance will be affected by transmitting frequency offset in dual-target scenario.}

{To attain a more thorough comprehension of the influences of carrier frequency offset, we conduct an additional simulation as shown in Fig.}~\ref{fig60} {. MUSIC-R algorithm only uses the phase differences among the rows in $\boldsymbol{Y}$ to estimate the target DOA. It can be seen that the estimation performance with MUSIC-R algorithm remains unchanged with different strength of transmitting frequency offset. This observation confirms that the phase differences among the rows are not affected by transmitting frequency offset, corroborating Proposition} \ref{prop1} {. Moreover, in dual-target scenarios, MUSIC-R algorithm demonstrates superior performance compared to 2D-MUSIC algorithm. The comparison suggests that relying solely on the phase differences between rows for DOA estimation can be more effective than co-estimation algorithms in the scenario with multiple targets and transmitting frequency offset. 
}

\begin{figure}[h]
\centering
\includegraphics[width=3.6 in]{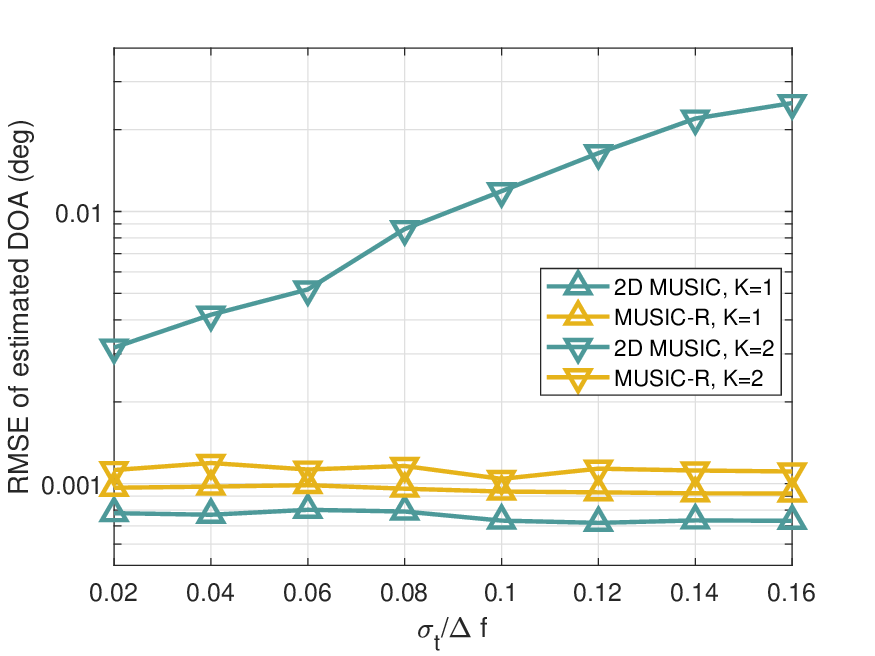}
\caption{Influence of transmitting frequency offset on estimation.}
\label{fig60}
\end{figure}

\begin{figure}[h]
\subfloat[Range estimation]{
\label{fig7a}
\includegraphics[width=3.6in]{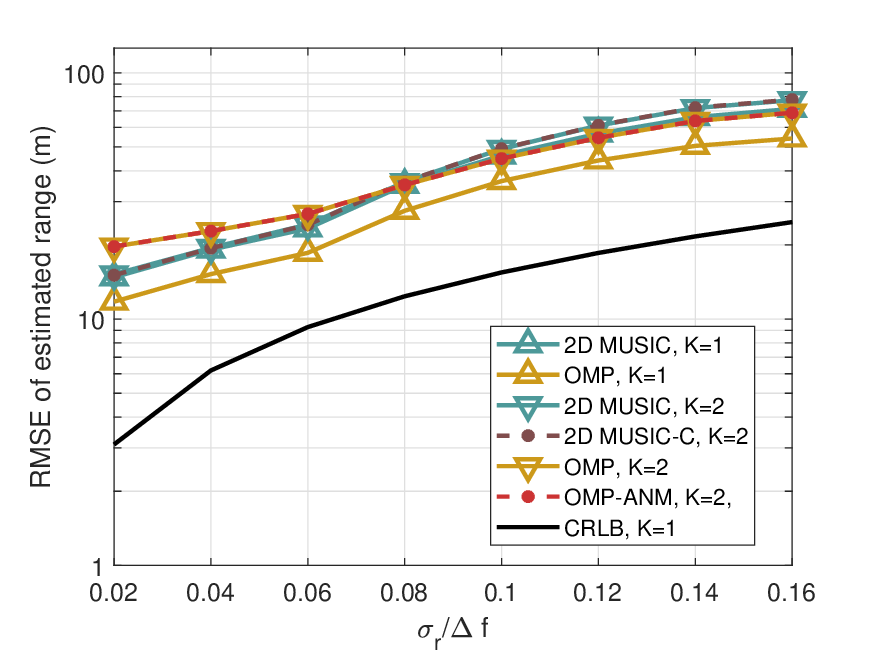}}
\quad
\subfloat[DOA estimation]{
\label{fig7b}  
\includegraphics[width=3.6 in]{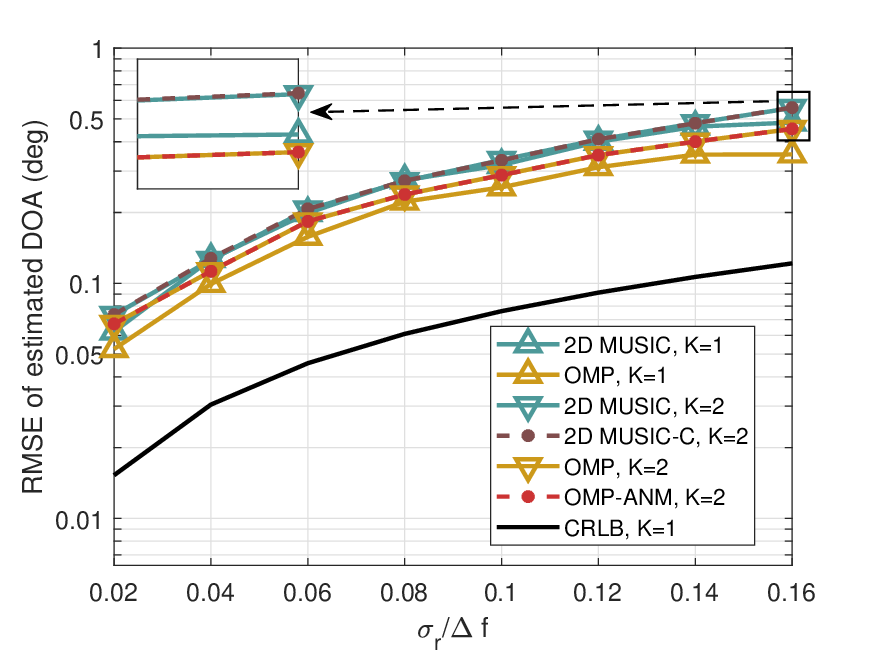}}
\caption{Influence of receiving frequency offset on estimation.}
\label{fig7}
\end{figure}

\begin{figure}[h]
\subfloat[SNR = $10$dB]{
\label{fig01a}
\includegraphics[width=3.6in]{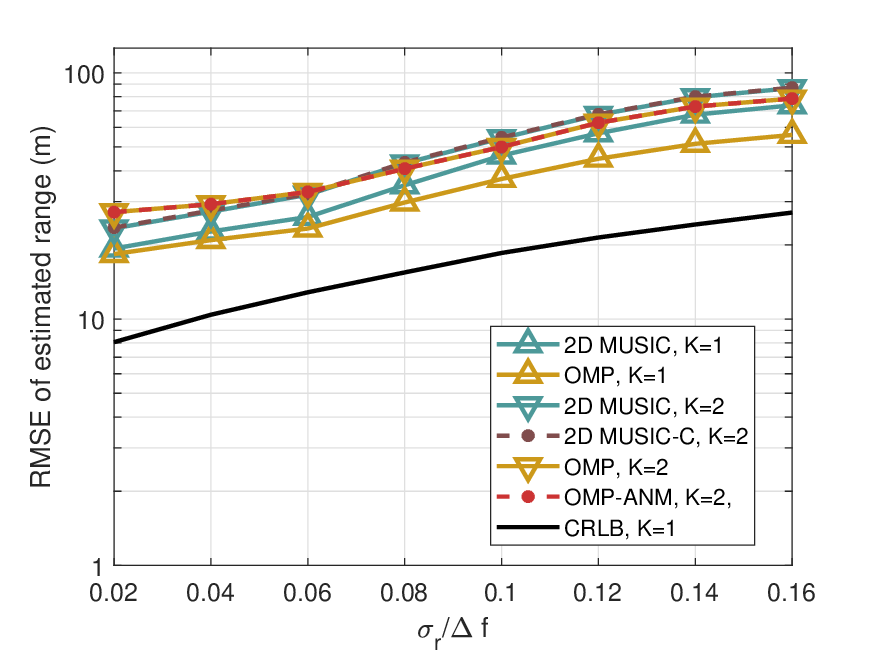}}
\quad
\subfloat[SNR = $0$dB]{
\label{fig01b}  
\includegraphics[width=3.6 in]{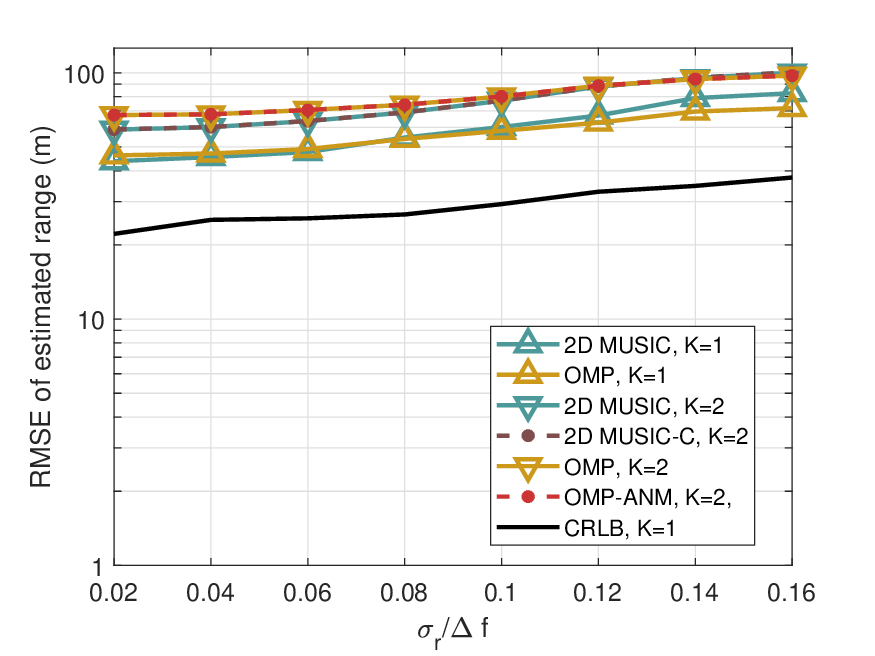}}
\caption{Influence of receiving frequency offset on range estimation with different SNR.}
\label{fig01}
\end{figure}
{Fig.}~\ref{fig7} {shows the deterioration of estimation caused by receiving frequency offsets. In single target scenario, OMP algorithm performs superior to MUSIC algorithms for both target range and DOA estimation. When the target number increases to $2$, OMP algorithm still outperforms MUSIC algorithm even with a severe performance deterioration. The performance deterioration of MUSIC algorithm caused by the target number increasing is comparatively limited. Moreover, similar to Fig.}~\ref{fig7}, { the introduction of the denoising algorithms, including both ANM and the fourth-order cumulant, do not bring a more precise estimation. Compared with the estimation performance in scenarios with carrier frequency offset, the performance deterioration caused by receiving frequency offset is basically equivalent. From the comparison between Fig.}~\ref{fig6} \subref{fig6a} {and Fig.}~\ref{fig7} \subref{fig7a}{, with the same standard deviation of the frequency offset, the RMSE of range estimation with receiving frequency offset is close to range estimation with transmitting frequency offset. Although the power of noise caused by the receiving frequency offset is higher than that caused by transmitting frequency offset, the colored noise also complicates the estimation.

{To compare the influence of frequency offsets under the scenarios with different strength of additive white noise, we set the SNR as $0$dB and $10$dB and show the range estimation performance in Fig.}~\ref{fig01}. {From the comparison among Fig.}~\ref{fig01} \subref{fig01a}{, Fig.}~\ref{fig01} \subref{fig01b} {and Fig.}~\ref{fig7} \subref{fig7a}{, the slope of the RMSE curves decrease with a lower SNR, which indicates a weaker influence of receiving frequency offset. Since each RMSE curve in the above figures represents the estimation performance with a varying strength of receiving frequency offset and a fixed strength of additive white noise, it is reasonable that the influence of frequency offset is comparatively decreased with a stronger additive white noise.}

\begin{figure}[h]
\subfloat[Range estimation]{
\label{fig8a}
\includegraphics[width=3.6in]{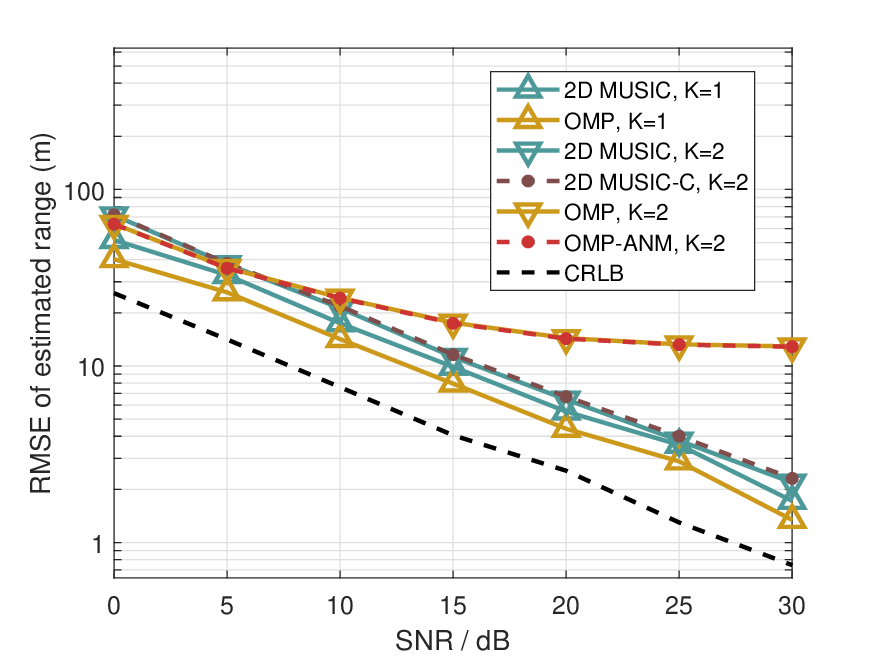}}
\quad
\subfloat[DOA estimation]{
\label{fig8b}  
\includegraphics[width=3.6 in]{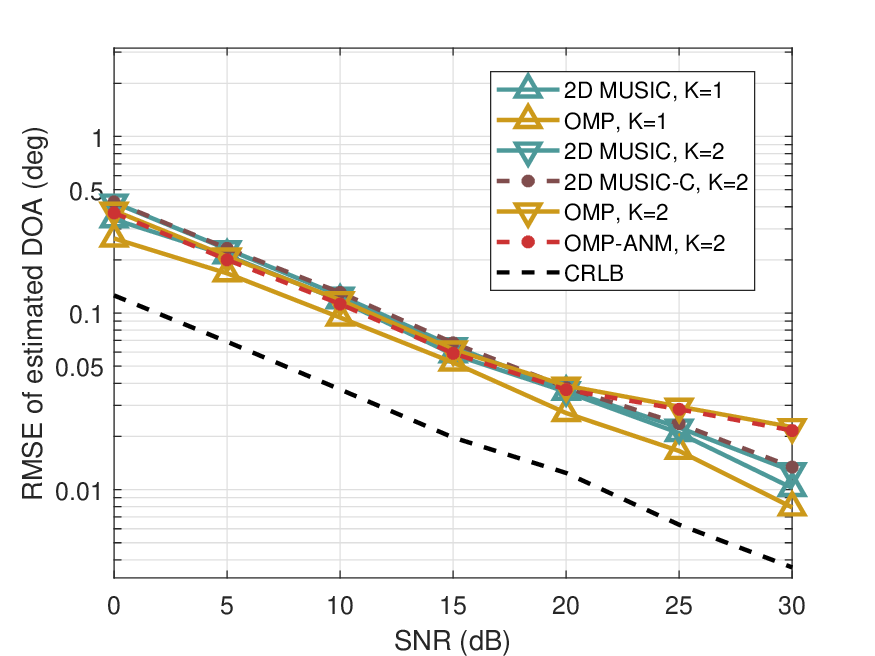}}
\caption{Influence of SNR on estimation.}
\label{fig8}
\end{figure}
{In general, our analysis reveals that algorithms based on OMP algorithm generally exhibit superior performance in single-target scenarios with transmitting frequency offsets. For multiple-target scenarios involving transmitting frequency offsets, MUSIC-R algorithm performs the best due to its highest target DOA estimation accuracy at the lowest computational complexity, which indicates that using only the phase differences among rows in the scenario is more effective compared with utilizing all phase information to obtain the DOA estimation. Besides, 2D MUSIC algorithm demonstrates the best range estimation performance in these scenarios.}
		
{In situations involving receiving frequency offsets, the OMP algorithm continues to excel in single-target scenarios, while the 2D-MUSIC algorithm proves to be the most effective in multiple-target scenarios. Moreover, it is noteworthy that the introduction of the ANM algorithm and the fourth-order cumulant, while offering marginal performance improvements, results in significantly higher computational complexity in scenarios with any type of frequency offsets.}

{We then show the relation between RMSE of the estimation and SNR in Fig.}~\ref{fig8}{. SNR here represents the ratio between the signal power and the white noise power, where the noise caused by frequency offsets is not included. The range and angle of the target are unchanged, and the frequency offsets are not attached. The SNR is defined at the matched filter outputs, i.e., the ratio of the power of the signal in matrix $\boldsymbol{Y}$ and the power of the noise in matrix $\boldsymbol{Y}$ in Eq.(}\ref{eq16}{). From the comparison between Fig.}~\ref{fig7} { and Fig.}~\ref{fig8}{, we can find that the deterioration caused by receiving frequency offset is basically in accordance with the analysis in Fig.}~\ref{fig4}{. For example, when the standard deviation of the receiving frequency offset $\sigma_r = 0.04 \Delta f$, the equalized SNR is about $11dB$ as shown in Fig.}~\ref{fig4}{. According to Fig.}~\ref{fig8}\subref{fig8a}{, the RMSE of estimated range using OMP algorithm in single target scenario are about $16m$ at $\textit{SNR} = 11dB$, and the RMSE of the estimated range is about $15.1m$ at $\sigma_r = 0.04 \Delta f$ as shown in Fig.}~\ref{fig7}{. The RMSE of range are basically the same when $\textit{SNR} = 11dB$ or $\sigma_r = 0.04 \Delta f$, which indicates the consistence between the simulation result and the analysis. Moreover, from the comparison between Fig.}~\ref{fig6} { and Fig.}~\ref{fig8}{, the deterioration on estimation caused by transmitting frequency offset is more severe than the the deterioration at the corresponding equalized SNR. The observation aligns with Proposition.} \ref{prop2}{, which points out that the equalized noise caused by transmitting frequency offset is colored. Compared with common additional white noise, the colored noise cause more severe degradation on estimation performance.  
}

\section{conclusion}\label{sec6}
The transmitting and receiving frequency offsets have been considered in the FDA-MIMO system. The model containing the frequency offsets has been built, and CRLB of estimation of range and angle are obtained. We have derived the equalized noise caused by frequency offsets and have analyzed the characteristics of the noise, together with their influence on the estimation of range and angle. {We have proved that the effect of the transmitting frequency offset is similar to additional colored noise and have utilized different algorithms to confirm the derivation}. Simulation results prove the analysis and numerically present the deterioration caused by frequency offsets. Future work will focus on applying the FDA radar in the ISAC system.

\bibliographystyle{IEEEtran}
\bibliography{ref}

\begin{thebibliography}{10}
\providecommand{\url}[1]{#1}
\csname url@samestyle\endcsname
\providecommand{\newblock}{\relax}
\providecommand{\bibinfo}[2]{#2}
\providecommand{\BIBentrySTDinterwordspacing}{\spaceskip=0pt\relax}
\providecommand{\BIBentryALTinterwordstretchfactor}{4}
\providecommand{\BIBentryALTinterwordspacing}{\spaceskip=\fontdimen2\font plus
\BIBentryALTinterwordstretchfactor\fontdimen3\font minus
  \fontdimen4\font\relax}
\providecommand{\BIBforeignlanguage}[2]{{%
\expandafter\ifx\csname l@#1\endcsname\relax
\typeout{** WARNING: IEEEtran.bst: No hyphenation pattern has been}%
\typeout{** loaded for the language `#1'. Using the pattern for}%
\typeout{** the default language instead.}%
\else
\language=\csname l@#1\endcsname
\fi
#2}}
\providecommand{\BIBdecl}{\relax}
\BIBdecl

\bibitem{sensing}
J.~Xu, G.~Liao, S.~Zhu, L.~Huang, and H.~C. So, ``Joint range and angle
  estimation using {MIMO} radar with frequency diverse array,'' \emph{IEEE
  Transactions on Signal Processing}, vol.~63, no.~13, pp. 3396--3410, 2015.

\bibitem{mainlobe}
W.-Q. Wang, H.~C. So, and A.~Farina, ``{FDA-MIMO} signal processing for
  mainlobe jammer suppression,'' in \emph{2019 27th European Signal Processing
  Conference (EUSIPCO)}, 2019, pp. 1--5.

\bibitem{CAESAR}
T.~Huang, N.~Shlezinger, X.~Xu, D.~Ma, Y.~Liu, and Y.~C. Eldar, ``Multi-carrier
  agile phased array radar,'' \emph{IEEE Transactions on Signal Processing},
  vol.~68, pp. 5706--5721, 2020.

\bibitem{FDA}
P.~Antonik, M.~Wicks, H.~Griffiths, and C.~Baker, ``Frequency diverse array
  radars,'' in \emph{2006 IEEE Conference on Radar}, 2006, pp. 3 pp.--.

\bibitem{FDA_clutter}
P.~Baizert, T.~Hale, M.~Temple, and M.~Wicks, ``Forward-looking radar gmti
  benefits using a linear frequency diverse array,'' \emph{Electronics
  Letters}, vol.~42, pp. 1311 -- 1312, 02 2006.

\bibitem{log_increment}
W.~Khan, I.~M. Qureshi, and S.~Saeed, ``Frequency diverse array radar with
  logarithmically increasing frequency offset,'' \emph{IEEE Antennas and
  Wireless Propagation Letters}, vol.~14, pp. 499--502, 2015.

\bibitem{random_increment}
W.~Wu and F.~Xi, ``Target localization for {FDA-MIMO} radar with random
  frequency increment via atomic norm minimization,'' in \emph{2019 IEEE MTT-S
  International Microwave Biomedical Conference (IMBioC)}, vol.~1, 2019, pp.
  1--4.

\bibitem{double_pulse}
W.-Q. Wang and H.~Shao, ``Range-angle localization of targets by a double-pulse
  frequency diverse array radar,'' \emph{IEEE Journal of Selected Topics in
  Signal Processing}, vol.~8, no.~1, pp. 106--114, 2014.

\bibitem{FDA_sbarray}
W.-Q. Wang, ``Subarray-based frequency diverse array radar for target
  range-angle estimation,'' \emph{IEEE Transactions on Aerospace and Electronic
  Systems}, vol.~50, no.~4, pp. 3057--3067, 2014.

\bibitem{FDA_MIMO}
P.~F. Sammartino, C.~J. Baker, and H.~D. Griffiths, ``Frequency diverse {MIMO}
  techniques for radar,'' \emph{IEEE Transactions on Aerospace and Electronic
  Systems}, vol.~49, no.~1, pp. 201--222, 2013.

\bibitem{FDA_MIMO2}
Y.~Zhu, L.~Liu, Z.~Lu, and S.~Zhang, ``Target detection performance analysis of
  {FDA-MIMO} radar,'' \emph{IEEE Access}, vol.~7, pp. 164\,276--164\,285, 2019.

\bibitem{LAN1}
L.~Lan, A.~Marino, A.~Aubry, A.~De~Maio, G.~Liao, J.~Xu, and Y.~Zhang,
  ``Glrt-based adaptive target detection in fda-mimo radar,'' \emph{IEEE
  Transactions on Aerospace and Electronic Systems}, vol.~57, no.~1, pp.
  597--613, 2021.

\bibitem{LAN2}
L.~Lan, J.~Xu, G.~Liao, Y.~Zhang, F.~Fioranelli, and H.~C. So, ``Suppression of
  mainbeam deceptive jammer with fda-mimo radar,'' \emph{IEEE Transactions on
  Vehicular Technology}, vol.~69, no.~10, pp. 11\,584--11\,598, 2020.

\bibitem{FDA_CRLB}
J.~Xiong, W.-Q. Wang, and K.~Gao, ``{FDA-MIMO} radar range–angle estimation:
  {CRLB}, {MSE}, and resolution analysis,'' \emph{IEEE Transactions on
  Aerospace and Electronic Systems}, vol.~54, no.~1, pp. 284--294, 2018.

\bibitem{FDA_MIMO_MUSIC}
M.~Feng, Z.~Cui, Y.~Yang, and Q.~Shu, ``A reduced-dimension {MUSIC} algorithm
  for monostatic {FDA-MIMO} radar,'' \emph{IEEE Communications Letters},
  vol.~25, no.~4, pp. 1279--1282, 2021.

\bibitem{OMP}
J.~A. Tropp and A.~C. Gilbert, ``Signal recovery from random measurements via
  orthogonal matching pursuit,'' \emph{IEEE Transactions on Information
  Theory}, vol.~53, no.~12, pp. 4655--4666, 2007.

\bibitem{FDA_MIMO_CS}
W.-G. Tang, H.~Jiang, and Q.~Zhang, ``Range-angle decoupling and estimation for
  {FDA-MIMO} radar via atomic norm minimization and accelerated proximal
  gradient,'' \emph{IEEE Signal Processing Letters}, vol.~27, pp. 366--370,
  2020.

\bibitem{CS_mismatch}
Y.~Li and Y.~Chi, ``Off-the-grid line spectrum denoising and estimation with
  multiple measurement vectors,'' \emph{IEEE Transactions on Signal
  Processing}, vol.~64, no.~5, pp. 1257--1269, 2016.

\bibitem{ANM}
Y.~Chi and M.~Ferreira Da~Costa, ``Harnessing sparsity over the continuum:
  Atomic norm minimization for superresolution,'' \emph{IEEE Signal Processing
  Magazine}, vol.~37, no.~2, pp. 39--57, 2020.

\bibitem{ANM_2d}
Y.~Chi and Y.~Chen, ``Compressive two-dimensional harmonic retrieval via atomic
  norm minimization,'' \emph{IEEE Transactions on Signal Processing}, vol.~63,
  no.~4, pp. 1030--1042, 2015.

\bibitem{DANM}
Z.~Zhang, Y.~Wang, and Z.~Tian, ``Efficient two-dimensional line spectrum
  estimation based on decoupled atomic norm minimization,'' \emph{Signal
  Processing}, vol. 163, 04 2019.

\bibitem{ANM_gp_error}
P.~Chen, Z.~Chen, Z.~Cao, and X.~Wang, ``A new atomic norm for {DOA} estimation
  with gain-phase errors,'' \emph{IEEE Transactions on Signal Processing},
  vol.~68, pp. 4293--4306, 2020.

\bibitem{PN}
K.~Siddiq, M.~K. Hobden, S.~R. Pennock, and R.~J. Watson, ``Phase noise in
  {FMCW} radar systems,'' \emph{IEEE Transactions on Aerospace and Electronic
  Systems}, vol.~55, no.~1, pp. 70--81, 2019.

\bibitem{fda_error}
K.~Gao, H.~Shao, H.~Chen, J.~Cai, and W.-Q. Wang, ``Impact of frequency
  increment errors on frequency diverse array {MIMO} in adaptive beamforming
  and target localization,'' \emph{Digital Signal Processing}, vol.~44, pp.
  58--67, 2015.

\bibitem{DANM2}
W.-G. Tang, H.~Jiang, and S.-X. Pang, ``Gridless angle and range estimation for
  {FDA-MIMO} radar based on decoupled atomic norm minimization,'' in
  \emph{ICASSP 2019 - 2019 IEEE International Conference on Acoustics, Speech
  and Signal Processing (ICASSP)}, 2019, pp. 4305--4309.

\bibitem{cn3}
A.~Swindlehurst and T.~Kailath, ``A performance analysis of subspace-based
  methods in the presence of model errors. {I}. the {MUSIC} algorithm,''
  \emph{IEEE Transactions on Signal Processing}, vol.~40, no.~7, pp.
  1758--1774, 1992.

\bibitem{cn5}
K.~Werner and M.~Jansson, ``{DOA} estimation and detection in colored noise
  using additional noise-only data,'' \emph{IEEE Transactions on Signal
  Processing}, vol.~55, no.~11, pp. 5309--5322, 2007.

\bibitem{cn4}
A.~Swindlehurst and T.~Kailath, ``A performance analysis of subspace-based
  methods in the presence of model error. {II}. multidimensional algorithms,''
  \emph{IEEE Transactions on Signal Processing}, vol.~41, no.~9, pp.
  2882--2890, 1993.

\bibitem{transformation}
B.~Friedlander and A.~Weiss, ``Direction finding using noise covariance
  modeling,'' \emph{IEEE Transactions on Signal Processing}, vol.~43, no.~7,
  pp. 1557--1567, 1995.

\bibitem{ML}
M.~Pesavento and A.~Gershman, ``Maximum-likelihood direction-of-arrival
  estimation in the presence of unknown nonuniform noise,'' \emph{IEEE
  Transactions on Signal Processing}, vol.~49, no.~7, pp. 1310--1324, 2001.

\bibitem{rotate}
A.~Paulraj and T.~Kailath, ``Eigenstructure methods for direction of arrival
  estimation in the presence of unknown noise fields,'' \emph{IEEE Transactions
  on Acoustics, Speech, and Signal Processing}, vol.~34, no.~1, pp. 13--20,
  1986.

\bibitem{four_level}
B.~Porat and B.~Friedlander, ``Direction finding algorithms based on high-order
  statistics,'' \emph{IEEE Transactions on Signal Processing}, vol.~39, no.~9,
  pp. 2016--2024, 1991.

\bibitem{whitening_transformation}
A.~C. Koivunen and A.~B. Kostinski, ``The feasibility of data whitening to
  improve performance of weather radar,'' \emph{Journal of Applied
  Meteorology}, vol.~38, no.~6, pp. 741 -- 749, 1999.

\bibitem{c4_1}
Z.~Liang, X.~Fu, and X.~Lv, ``A novel channel inconsistency calibration
  algorithm for azimuth multichannel sar based on fourth-order cumulant,''
  \emph{IEEE Journal of Selected Topics in Applied Earth Observations and
  Remote Sensing}, vol.~16, pp. 5561--5577, 2023.

\bibitem{c4_2}
W.-J. Zeng, X.-L. Li, and X.-D. Zhang, ``Direction-of-arrival estimation based
  on the joint diagonalization structure of multiple fourth-order cumulant
  matrices,'' \emph{IEEE Signal Processing Letters}, vol.~16, no.~3, pp.
  164--167, 2009.

\bibitem{E_N}
C.~Wen, Y.~Xie, Z.~Qiao, L.~Xu, and Y.~Qian, ``A tensor generalized weighted
  linear predictor for {FDA-MIMO} radar parameter estimation,'' \emph{IEEE
  Transactions on Vehicular Technology}, vol.~71, no.~6, pp. 6059--6072, 2022.

\bibitem{CFO_tr}
K.~Iqbal, J.~Ahmed, and A.~Rafique, ``Analysis of carrier frequency offset
  distribution on efficiency of multicarrier spread spectrum techniques,'' in
  \emph{2016 International Conference on Frontiers of Information Technology
  (FIT)}, 2016, pp. 125--129.

\bibitem{CFO_0}
J.~Ahmed and K.~A. Hamdi, ``Spectral efficiency of asynchronous {MC-CDMA} with
  frequency offset over correlated fading,'' \emph{IEEE Transactions on
  Vehicular Technology}, vol.~62, no.~7, pp. 3423--3429, 2013.

\bibitem{fet10}
K.~Iqbal, J.~Ahmed, and A.~Rafique, ``Analysis of carrier frequency offset
  distribution on efficiency of multicarrier spread spectrum techniques,'' in
  \emph{2016 International Conference on Frontiers of Information Technology
  (FIT)}, 2016, pp. 125--129.

\bibitem{FDA_orth}
B.~Yang, S.~Zhu, X.~He, L.~Lan, and X.~Li, ``Cognitive {FDA-MIMO} radar network
  for target discrimination and tracking with main-lobe deceptive trajectory
  interference,'' \emph{IEEE Transactions on Aerospace and Electronic Systems},
  vol.~PP, pp. 1--16, 08 2023.

\bibitem{CFO_r}
J.~Winkel, ``Modeling and simulating gnss signal structures and receivers,''
  Ph.D. dissertation, 11 2003.

\bibitem{colored_noise}
K.~Werner and M.~Jansson, ``{DOA} estimation and detection in colored noise
  using additional noise-only data,'' \emph{IEEE Transactions on Signal
  Processing}, vol.~55, no.~11, pp. 5309--5322, 2007.

\bibitem{colored_noise2}
B.~Goransson and B.~Ottersten, ``Direction estimation in partially unknown
  noise fields,'' \emph{IEEE Transactions on Signal Processing}, vol.~47,
  no.~9, pp. 2375--2385, 1999.

\bibitem{colored_noise3}
W.~Zuo, J.~Xin, N.~Zheng, H.~Ohmori, and A.~Sano, ``Subspace-based near-field
  source localization in unknown spatially nonuniform noise environment,''
  \emph{IEEE Transactions on Signal Processing}, vol.~68, pp. 4713--4726, 2020.

\bibitem{IST}
S.~Khoramian, ``An iterative thresholding algorithm for linear inverse problems
  with multi-constraints and its applications,'' \emph{Applied and
  Computational Harmonic Analysis}, vol.~32, no.~1, pp. 109--130, 2012.

\bibitem{ANM2}
B.~N. Bhaskar, G.~Tang, and B.~Recht, ``Atomic norm denoising with applications
  to line spectral estimation,'' \emph{IEEE Transactions on Signal Processing},
  vol.~61, no.~23, pp. 5987--5999, 2013.

\bibitem{ANM3}
Y.~Chi and M.~Ferreira Da~Costa, ``Harnessing sparsity over the continuum:
  Atomic norm minimization for superresolution,'' \emph{IEEE Signal Processing
  Magazine}, vol.~37, no.~2, pp. 39--57, 2020.

\bibitem{ANM4}
Y.~Chi, ``Guaranteed blind sparse spikes deconvolution via lifting and convex
  optimization,'' \emph{IEEE Journal of Selected Topics in Signal Processing},
  vol.~10, no.~4, pp. 782--794, 2016.

\bibitem{FIM_singular}
Z.~Ben-Haim and Y.~C. Eldar, ``On the constrained {C}ram\'{e}r-{R}ao bound with
  a singular fisher information matrix,'' \emph{IEEE Signal Processing
  Letters}, vol.~16, no.~6, pp. 453--456, 2009.

\bibitem{FIM_2}
B.-Z. Bobrovsky, E.~Mayer-Wolf, and M.~Zakai, ``Some classes of global
  {C}ram\'{e}r-{R}ao bounds,'' \emph{The Annals of Statistics}, vol.~15, 12
  1987.

\end{thebibliography}
\begin{IEEEbiography}
[{\includegraphics[width=1in,height=1.25in,clip,keepaspectratio]{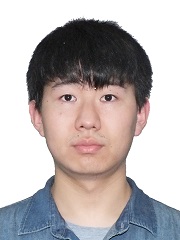}}] 
{Mengjiang Sun (Student Member, IEEE)} was born in Inner Mongolia, China, in 1998. He received the B.E. degree from School of Information Science and Engineering, Southeast University, China in 2021. He is currently pursing the Ph.D degree with the State Key Laboratory of Millimeter Waves, Southeast University, Nanjing, China.
His research interests include radar signal processing and millimeter wave communication.
\end{IEEEbiography}
\vspace{-90 mm} 
\begin{IEEEbiography}
[{\includegraphics[width=1in,height=1.25in,clip,keepaspectratio]{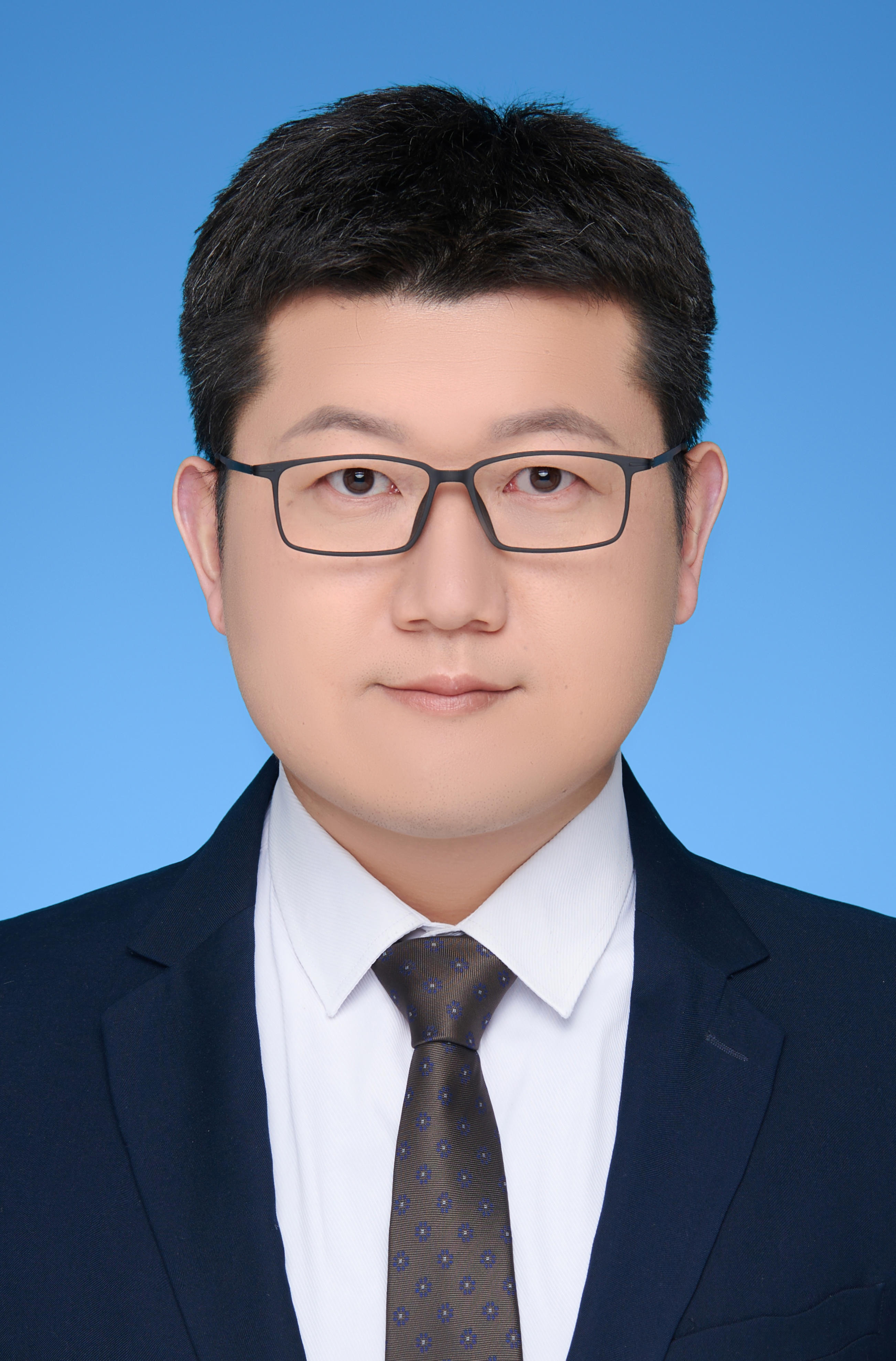}}] 
{Peng Chen (Seinor Member, IEEE)} received the B.E. and Ph.D. degrees from the School of Information Science and Engineering, Southeast University, Nanjing, China, in 2011 and 2017 respectively. From March 2015 to April 2016, he was a Visiting Scholar with the Department of Electrical Engineering, Columbia University, New York, NY, USA. He is currently an Associate Professor with the State Key Laboratory of Millimeter Waves, Southeast University. His research interests include target localization, super-resolution reconstruction, and array signal processing. He is a Jiangsu Province Outstanding Young Scientist. He has served as an IEEE ICCC Session Chair, and won the Best Presentation Award in 2022 (IEEE ICCC). He was invited as a keynote speaker at the IEEE ICET in 2022. He was recognized as an exemplary reviewer for IEEE WCL in 2021, and won the Best Paper Award at IEEE ICCCCEE in 2017.
\end{IEEEbiography}
\vspace{-94 mm} 
\begin{IEEEbiography}[{\includegraphics[width=1in,height=1.25in,clip,keepaspectratio]{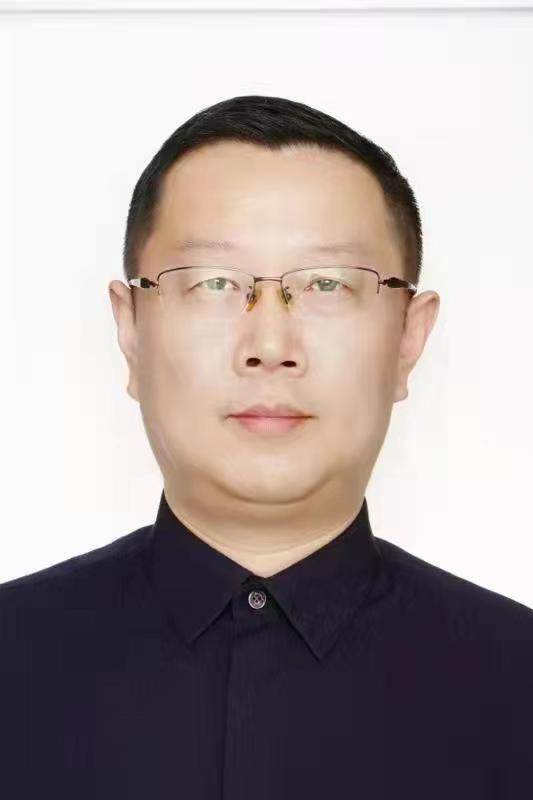}}]{Zhenxin Cao (Member, IEEE)} was born in May 1976. He received the M.S. degree from Nanjing University of Aeronautics and Astronautics, Nanjing, China, in 2002 and the Ph.D. degree from the School of Information Science and Engineering, Southeast University, Nanjing, China, in 2005. From 2012 to 2013, he was a Visiting Scholar with North Carolina State University. Since 2005, he has been with the State Key Laboratory of Millimeter Waves, Southeast University, where he is currently a Professor. His research interests include antenna theory and application.
\end{IEEEbiography}

\end{document}